\documentclass[opre,nonblindrev]{informs3-altered}

\SingleSpacedXI

\renewcommand{\footnoterule}{%
  \kern -3pt
  \hrule width .2\textwidth height 1pt
  \kern 2pt
}

\usepackage{natbib}
 \bibpunct[, ]{[}{]}{,}{a}{}{,}%
 \def\bibfont{\small}%
\TheoremsNumberedBySection
\ECRepeatTheorems

\EquationsNumberedBySection 

\usepackage[]{hyperref}
\hypersetup{
    unicode=false,          
    pdftoolbar=true,        
    pdfmenubar=true,        
    pdffitwindow=false,      
    pdftitle={Bayesian Incentive-Compatible Bandit Exploration},    
    pdfauthor={Mansour, Slivkins and Syrgkanis}
    colorlinks=false,       
}

\usepackage[ruled,norelsize]{algorithm2e}

\newcommand{\algTab}{\hspace{5mm}}

\usepackage{slivkins-setup}
\usepackage{vasilis-setup}
\renewcommand{\Pr}[2][]{\ensuremath{\mathop{\rm Pr}_{#1}\left(#2\right)}}

\newcommand{\BetaD}{\ensuremath{\mathtt{Beta}}\xspace}

\newcommand{\OurTitle}{Bayesian Incentive-Compatible Bandit Exploration}
\newcommand{\propref}[1]{Property~\refprop{#1}}

\newcommand{\ALG}{\ensuremath{\mA}\xspace}
\newcommand{\ALGIC}{\ensuremath{\ALG^{\mathtt{IC}}}\xspace}

\renewenvironment{proof}[1][Proof]
{\Trivlist\item[\hspace*{1em}\hskip\labelsep{\it #1.\enskip}]\ignorespaces}
{{\hfill \Halmos}\endTrivlist\addvspace{2pt}}

\begin{document}


\RUNAUTHOR{Mansour, Slivkins and Syrgkanis}

\RUNTITLE{\OurTitle}

\TITLE{\OurTitle}

\ARTICLEAUTHORS{%
\AUTHOR{Yishay Mansour%
    \footnote{Most of this research has been done while Y. Mansour was a researcher at Microsoft Research, Herzelia, Israel.}}
\AFF{Tel Aviv University, Tel Aviv, Israel, \EMAIL{mansour.yishay@gmail.com}}
\AUTHOR{Aleksandrs Slivkins}
\AFF{Microsoft Research,  New York, NY 10011, USA, \EMAIL{slivkins@microsoft.com}}
\AUTHOR{Vasilis Syrgkanis}
\AFF{Microsoft Research, Cambridge, MA 02142, USA, \EMAIL{vasy@microsoft.com}}
} 

\ABSTRACT{%
As self-interested individuals ("agents") make decisions over time, they utilize information revealed by other agents in the past, and produce information that may help agents in the future. This phenomenon is common in a wide range of scenarios in the Internet economy, as well as in medical decisions. Each agent would like to \emph{exploit}: select the best action given the current information, but would prefer the previous agents to \emph{explore}: try out various alternatives to collect information. A social planner, by means of a carefully designed recommendation policy, can incentivize the agents to balance the exploration and exploitation so as to maximize social welfare.

We model the planner's recommendation policy as a multi-arm bandit algorithm under incentive-compatibility constraints induced by agents' Bayesian priors. We design a bandit algorithm which is incentive-compatible and has asymptotically optimal performance, as expressed by regret. Further, we provide a black-box reduction from an arbitrary multi-arm bandit algorithm to an incentive-compatible one, with only a constant multiplicative increase in regret. This reduction works for very general bandit setting that incorporate contexts and arbitrary partial feedback.

\OMIT{ OLD ABSTRACT
Individual decision-makers consume information revealed by the previous decision makers, and produce information that may help in future decisions. This phenomenon is common in a wide range of scenarios in the Internet economy, as well as in other domains such as medical decisions. Each decision-maker would individually prefer to \emph{exploit}: select an action with the highest expected reward given her current information. At the same time, each decision-maker would prefer previous decision-makers to \emph{explore}, producing information about the rewards of various actions. A social planner, by means of carefully designed information disclosure, can incentivize the agents to balance the exploration and exploitation so as to maximize social welfare.

We formulate this problem as a multi-armed bandit problem (and various generalizations thereof) under incentive-compatibility constraints induced by the agents' Bayesian priors. We design an  incentive-compatible bandit algorithm for the social planner whose regret is asymptotically optimal among all bandit algorithms (incentive-compatible or not). Further, we provide a black-box reduction from an arbitrary multi-arm bandit algorithm to an incentive-compatible one, with only a constant multiplicative increase in regret. This reduction works for very general bandit setting that incorporate contexts and arbitrary auxiliary feedback.
} 
}%


\HISTORY{First version: February 2015. This version: May 2019.\footnote{An extended abstract of this paper has been published in \emph{ACM EC 2015} (16th ACM Conf. on Economics and Computation). Compared to the version in conference proceedings, this version contains complete proofs, revamped introductory sections, and thoroughly revised presentation of the technical material. Further, two major extensions are fleshed out, resp. to more than two actions and to more general machine learning settings, whereas they were only informally described in the conference version. The main results are unchanged, but their formulation and presentation is streamlined, particularly regarding assumptions on the common prior. This version also adds a discussion of potential applications to medical trials.}}

\maketitle

\newpage
\tableofcontents
\vfill
\newpage


\section{Introduction}
\label{sec:intro}

Decisions made by an individual often reveal information about the world that can be useful to others. For example, the decision to dine in a particular restaurant may reveal some observations about this restaurant. This revelation could be achieved, for example, by posting a photo, tweeting, or writing a review. Others can consume this information either directly (via photo, review, tweet, etc.) or indirectly through aggregations, summarizations or recommendations. Thus, individuals have a dual role: they both consume information from previous individuals and produce information for future consumption. This phenomenon applies very broadly: the choice of a product or experience, be it a movie, hotel, book, home appliance, or virtually any other consumer's choice, leads to an individual's subjective observations pertaining to this choice. These subjective observations can be recorded and collected, e.g., when the individual ranks a product or leaves a review, and can help others make similar choices in similar circumstances in a more informed way. Collecting, aggregating and presenting such observations is a crucial value proposition of numerous businesses in the modern Internet economy, such as TripAdvisor, Yelp, Netflix, Amazon, Waze and many others (see \reftbl{recommendation-systems}). Similar issues, albeit possibly with much higher stakes, arise in medical decisions:
selecting a doctor or a hospital, choosing a drug or a treatment,
or deciding whether to participate in a medical trial.
First the individual can consult information from similar individuals in the past, to the extent that such information is available, and later he can contribute her experience as a review or as an outcome in a medical trial.

\begin{table}[h]
\begin{center}
\begin{tabular}{l|c}
Watch this movie        & Netflix  \\
Dine in this restaurant & Yelp \\
Vacation in this resort & TripAdvisor \\
Buy this product        & Amazon \\
Drive this route        & Waze \\ \hline
Do this exercise        & FitBit \\
See this doctor         & SuggestADoctor \\
Take this medicine      & medical trial
\end{tabular}
\end{center}
\caption{Systems for recommendations and collecting feedback.}
\label{tbl:recommendation-systems}
\end{table}

If a social planner were to direct the individuals in the information-revealing decisions discussed above, she would have two conflicting goals: \emph{exploitation}, choose the best alternative given the information available so far, and \emph{exploration}, trying out less known alternatives for the sake of gathering more information, at the risk of worsening the individual experience. A social planner would like to combine exploration and exploitation so as to maximize the social welfare, which results in the \emph{exploration-exploitation tradeoff}, a well-known subject in Machine Learning, Operation Research and Economics.

However, when the decisions are made by individuals rather than enforced by the planner, we have another problem dimension based on the individuals' incentives. While the social planner benefits from both exploration and exploitation, each individuals' incentives are typically skewed in favor of the latter. (In particular, many people prefer to benefit from exploration done by others.) Therefore, the society as a whole may suffer from insufficient amount of exploration. In particular, if a given alternative appears suboptimal given the information available so far, however sparse and incomplete, then this alternative may remain unexplored -- even though it may be the best.

The focus of this work is how to incentivize self-interested decision-makers to explore. We consider a social planner who cannot control the decision-makers, but can communicate with them, e.g., recommend an action and observe the outcome later on.
Such a planner would typically be implemented via a website, either one dedicated to recommendations and feedback collection (such as Yelp or Waze), or one that actually provides the product or experience being recommended (such as Netflix or Amazon). In medical applications, the planner would be either a website that rates/recommends doctors and collects reviews on their services
(such as \emph{rateMDs.com} or \emph{SuggestADoctor.com}), or an organization conducting a medical trial. We are primarily interested in exploration that is efficient from the social planner's perspective, \ie exploration that optimizes the social welfare.%
\footnote{In the context of Internet economy, the ``planner" would be a for-profit company. Yet, the planner's goal, for the purposes of incentivinzing exploration, would typically be closely aligned with the social welfare.}

Following \citet{Kremer-JPE14}, we consider a basic scenario when the only incentive offered by the social planner is the recommended experience itself (or rather, the individual's belief about the expected utility of this experience). In particular, the planner does not offer payments for following the recommendation.
On a technical level, we study a mechanism design problem with an explore-exploit tradeoff and auxiliary incentive-compatibility constraints. Absent these constraints, our problem reduces to \emph{multi-armed bandits} (\emph{MAB}) with stochastic rewards, the paradigmatic and well-studied setting for exploration-exploitation tradeoffs, and various generalizations thereof. The interaction between the planner and a single agent can be viewed as a version of the \emph{Bayesian Persuasion} game \citep{Kamenica-aer11} in which the planner has more information due to the feedback from the previous agents; in fact, this \emph{information asymmetry} is crucial for ensuring the desired incentives.


\subsection{Our model and scope}
\label{sec:intro-scope}

We consider the following abstract framework, called \emph{incentive-compatible exploration}. The social planner is an algorithm that interacts with the self-interested decision-makers (henceforth, \emph{agents}) over time. In each round, an agent arrives, chooses one action among several alternatives, receives a reward for the chosen action, and leaves forever. Before an agent makes her choice, the planner sends a message to the agent which includes a recommended action.
Everything that happens in a given round is observed by the planner, but not by other agents. The agent has a Bayesian prior on the reward distribution, and chooses an action that maximizes its Bayesian expected reward given the algorithm's message (breaking ties in favor of the recommended action). The agent's prior is known to the planner, either fully or partially. We require the planner's algorithm to be \emph{Bayesian incentive-compatible} (henceforth, \emph{BIC}), in the sense that each agent's Bayesian expected reward is maximized by the recommended action. The basic goal is to design a BIC algorithm so as to maximize social welfare, \ie the cumulative reward of all agents.

The algorithm's message to each agent is restricted to the recommended action (call such algorithms \emph{message-restricted}). Any BIC algorithm can be turned into a message-restricted BIC algorithm which chooses the same actions, as long as the agents' priors are exactly known to the planner.%
\footnote{This is by ``revelation principle", as observed in \citet{Kremer-JPE14} (see Remark~\ref{rem:revelation-principle} in the Preliminaries).}
Note that a message-restricted algorithm (BIC or not) is simply an MAB-like learning algorithm for the same setting.

A paradigmatic example is the setting where the reward is an independent draw from a distribution determined only by the chosen action. All agents share a common Bayesian prior on the reward distribution; the prior is also known to the planner. No other information is received by the algorithm or an agent (apart from the prior, the recommended action, and the reward for the chosen action).
We call this setting \emph{BIC bandit exploration}. Absent the BIC constraint, it reduces to the MAB problem with IID rewards and Bayesian priors. We also generalize this setting in several directions, both in terms of the machine learning problem being solved by the planner's algorithm, and in terms of the mechanism design assumptions on the information structure.

We assume that each agent knows which round he is arriving in. This assumption is without loss of generality, because an algorithm that is BIC under this assumption is also BIC for a more general version in which each agent has a Bayesian belief about his round. (Also, the round can simply be included in the planner's message to this agent.)

\xhdr{Discussion.}
BIC exploration can be interpreted as a protection from the  \emph{selection bias}, when the population that participates in the experiment differs from the target population, potentially affecting the findings. More specifically, both action and the observed outcome may depend on agents' properties, \eg the people who rate a particular niche movie are the people who chose to see this movie, and therefore (this being a niche movie) are much more likely to enjoy it.

BIC exploration does not rely on ``external" incentives such as monetary payments or discounts, social status distinctions (\eg leaderboards or merit badges for prolific feedback contributors), or people's affinity towards experimentation. Such approaches can themseleves cause some selection bias. Indeed, paying patients for participation in a medical trial may be more appealing to poorer patients; offering discounts for new services may attract customers who are more sensitive to such discounts; and relying on people who like to explore for themselves would lead to a dataset that represents this category of people rather than the general population.  While all these approaches are reasonable and in fact widely used (with well-developed statistical tools to mitigate the selection bias), an alternative intrinsically less prone to selection bias is, in our opinion, worth investigating.

The ``intrinsic" incentives offered by BIC exploration can be viewed as a guarantee of \emph{fairness} for the agents. Indeed, even though the planner imposes experimentation on the agents, the said experimentation does not degrade expected utility of any one agent. This is simply because an agent can always choose to ignore the planner's recommendation and select an action with the highest prior mean reward. This is particularly important for settings in which ``external" incentives described above do not fully substitute for low intrinsic utility of the chosen actions. For example, a monetary payment does not fully substitute for an adverse outcome of a medical trial, and a discounted meal at a restaurant does not fully substitute for a bad experience.

However, the BIC property is stronger than the fairness property discussed above. Indeed, following the recommendation should be compared not only to ignoring it (and going with the best action according to the prior), but also to choosing some other action \emph{given the recommendation}. In other words, an agent may see a particular action being recommended, do a Bayesian update on this fact, and decide to choose a different action; BIC property does not allow this to happen.

We focus on message-restricted algorithms, and rely on the BIC property to convince agents to follow our recommendations. We do not attempt to make our recommendations more convincing by revealing additional information, because doing so does not help in our model, and because the desirable kinds of additional information to be revealed are likely to be application-specific (whereas with message-restricted algorithms we capture many potential applications at once). Further, message-restricted algorithms are allowed, and even recommended, in the domain of medical trials (see Section~\ref{sec:medical} for discussion).

We make the standard assumptions of trust and rationality. We assume the ``power of commitment": the planner can commit to a particular policy, and the agents trust that the planner actually implements it. We posit that the agents are Bayesian and risk-neutral.%
\footnote{A version of risk-aversion may be modeled in a standard way, by redefining rewards for the undesirable outcomes.}
Finally, we assume that the agents either trust that our mechanisms are Bayesian incentive-compatible, or are able to verify this property themselves or via third party. While non-trivial, such assumptions are standard and tremendously useful in many scenarios throughout economics.

\xhdr{Objectives.} We seek BIC algorithms whose performance is near-optimal for the corresponding setting without the BIC constraint. This is a common viewpoint for welfare-optimizing mechanism design problems, which often leads to strong positive results, both prior-independent and Bayesian, even if Bayesian-optimal BIC algorithms are beyond one's reach. Prior-independent guarantees are particularly desirable because priors are almost never completely correct in practice.

We express prior-independent performance of an algorithm via a standard notion of \emph{regret}: the difference, in terms of the cumulative expected reward, between the algorithm and the the best fixed action. Intuitively, it is the extent to which the algorithm ``regrets" not knowing the best action in advance. For Bayesian performance, we consider \emph{Bayesian regret}: ex-post regret in expectation over the prior, and also the average Bayesian-expected reward per agent. (For clarity, we will refer to the prior-independent version as \emph{ex-post regret}.) Moreover, we consider a version in which the algorithm outputs a prediction after each round (visible only to the planner), e.g., the predicted best action; then we are interested in the rate at which this prediction improves over time.

\subsection{Our contributions}
\label{sec:contributions}

On a high level, we make two contributions:
\vspace{2mm}
\begin{description}
\item{\emph{Regret minimization}}
We provide an algorithm for BIC bandit exploration whose ex-post regret is asymptotically optimal among all MAB algorithms (assuming a constant number of actions). Our algorithm is \emph{detail-free}, in that it requires only a limited knowledge of the prior.
\vspace{2mm}
\item{\emph{Black-box reduction}} We provide a reduction from an arbitrary learning algorithm to a BIC one, with only a minor loss in performance; this reduction ``works" for a very general setting.
\end{description}
\vspace{2mm}
In what follows we discuss our results in more detail.

\OMIT{
the rich field of Algorithmic Mechanism Design \ascomment{what to cite?}, whenever algorithms for the same setting and the same objective are well-defined.%
\footnote{This is the case for many welfare-optimization problems, but not for most revenue-maximization problems.}
is essentially a holy grail of Algorithmic Mechanism Design, as applied to a particular application domain. \ascomment{does ``holy grail" sound OK here? citation?} Armed with a reduction such as
} 

\xhdr{Regret minimization.}
Following the literature on regret minimization, we focus on the asymptotic ex-post regret rate as a function of the time horizon (which in our setting corresponds to the number of agents). We establish that the BIC restriction does not affect the asymptotically optimal ex-post regret rate. The optimality is two-fold: in the worst case over all realizations of the common prior (\ie for every possible vector of expected rewards), and for every particular realization of the prior (which may allow much smaller ex-post regret than in the worst-case). More formally, if $T$ is the time horizon and $\Delta$ is the ``gap" in expected reward between the best action and the second-best action, then our algorithm achieves ex-post regret
\begin{align}\label{eq:intro-MAB-regret}
    c_\mP + c_0\cdot\min(\tfrac{m}{\Delta} \log T, \sqrt{mT\log T}),
\end{align}
where $c_0$ is an absolute constant, and $c_\mP<0$ depends only on the common prior $\mP$. Without the BIC restriction, bandit algorithms achieve a similar regret bound, but without the $c_\mP$ term \citep{bandits-ucb1}. No bandit algorithm, BIC or not, can achieve ex-post regret better than
     $O(\min(\tfrac{m}{\Delta} \log T,\; \sqrt{mT}))$
     \citep{Lai-Robbins-85,bandits-exp3}.
\footnote{More precisely, this is the best regret bound that can be achieved for all problem instances with gap $\Delta$.}

\OMIT{The per-realization result is expressed in terms of the time horizon and the ``gap" between the best and second-best action.}

\OMIT{Third, the per-realization regret rate quantifies the advantage of ``nice" problem instances (those with a large ``gap"), and of stochastic rewards vs. rewards selected in advance by an adversary.}

Conceptually, our algorithm implements \emph{adaptive exploration}: the exploration schedule is adapted to the observations, so that low-performing actions are phased out early. This is known to be a vastly superior approach compared to exploration schedules that are fixed in advance.%
\footnote{For multi-armed bandit algorithms with fixed exploration schedules, regret can be no better than $\Omega(m^{1/3}\,T^{2/3})$ in the worst case over all problem instances, and (essentially) no better than $O(\Delta T^{\Omega(1)})$ in the worst case over all problem instances with gap $\Delta$ \citep{MechMAB-ec09}.}

Further, our algorithm is \emph{detail-free}, requiring only a limited knowledge of the common prior. This is desirable because in practice it may be complicated or impossible to elicit the  prior exactly. Moreover, this feature allows the agents to have different priors (as long as they are ``compatible" with the planner's prior, in a precise sense specified later). In fact, an agent does not even need to know her prior exactly: instead, she would trust the planner's recommendation as long as she believes that their priors are compatible.

The prior-dependent additive term $c_\mP$ in \eqref{eq:intro-MAB-regret} can get arbitrarily large depending on the prior, even for two arms. Informally, the magnitude of $c_\mP$ is proportional to how difficult it is to convince an agent to try a seemingly suboptimal action. This term can also grow exponentially with the number of arms $m$ in some examples. While (very) suboptimal for multi-armed bandits, we conjecture that such dependence is an inevitable ``price of incentive-compatibility".

\xhdr{Black-box reduction.}
Given an arbitrary MAB algorithm $\ALG$, we provide a BIC algorithm $\ALGIC$ which internally uses $\ALG$ as a ``black-box". That is, $\ALGIC$ simulates a run of $\ALG$, providing inputs and recording the respective outputs, but does not depend on the internal workings of $\ALG$. In addition to recommending an action, the original algorithm $\ALG$ can also output a \emph{prediction} after each round (visible only to the planner), e.g. the predicted best action; then $\ALGIC$ outputs a prediction, too. A reduction such as ours allows a modular design: one can design a non-BIC algorithm (or take an existing one), and then use the reduction to inject incentive-compatibility. Modular designs are very desirable in complex economic systems, especially for settings such as MAB with a rich body of existing work.

Our reduction incurs only a small loss in performance, which can be quantified in several ways. In terms of Bayesian regret, the performance of $\ALGIC$ worsens by at most a constant multiplicative factor that only depends on the prior. In terms of the average rewards, we guarantee the following: for any duration $T$, the average Bayesian-expected reward of $\ALGIC$ between rounds $c_\mP$ and $c_\mP+T$ is at least that of the first $T/L_\mP$ rounds in the original algorithm $\ALG$; here $c_\mP$ and $L_\mP$ are prior-dependent constants.%
\footnote{As in \eqref{eq:intro-MAB-regret}, the prior-dependent constants can get arbitrarily large depending on the prior, even for $m=2$ arms, and can scale exponentially with $m$ in some examples.}
Finally, if $\ALG$ outputs a prediction $\phi_t$ after each round $t$, then $\ALGIC$ learns as fast as $\ALG$, up to a prior-dependent constant factor $c_\mP$: for every realization of the prior, its prediction in round $t$ has the same distribution as $\phi_{\flr{t/c_\mP}}$.%
\footnote{So if the original algorithm $\ALG$ gives an asymptotically optimal error rate as a function of $t$, compared to the ``correct" prediction, then so does the transformed algorithm $\ALGIC$, up to a prior-dependent multiplicative factor.}

The black-box reduction has several benefits other than ``modular design". Most immediately,
one can plug in an MAB algorithm that takes into account the Bayesian prior or any other auxiliary information that a planner might have. Moreover, one may wish to implement a particular approach to exploration, e.g., incorporate some constraints on the losses, or preferences about which arms to favor or to throttle. Further, the planner may wish to predict things other than the best action. To take a very stark example, the planner may wish to learn what are the \emph{worst} actions (in order to eliminate these actions later by other means such as  legislation). While the agents would not normally dwell on low-performing actions, our reduction would then incentivize them to explore these actions in detail.

\xhdr{Beyond BIC bandit exploration.}
Our black-box reduction supports much richer scenarios than BIC \emph{bandit} exploration. Most importantly, it allows for agent heterogeneity, as expressed by observable signals. We adopt the framework of \emph{contextual bandits}, well-established in the machine learning literature (see Section~\ref{sec:related-work} for citations). In particular, each agent is characterized by a signal, called \emph{context}, observable by both the agent and the planner before the planner issues the recommendation. The context can include demographics, tastes, preferences and other agent-specific information. It impacts the expected rewards received by this agent, as well as the agent's beliefs about these rewards. Rather than choose the best action, the planner now wishes to optimize a \emph{policy} that maps contexts to actions. This type of agent heterogeneity is practically important: for example, websites that issue recommendations may possess a huge amount of information about their customers, and routinely use this ``context" to adjust their recommendations (e.g., Amazon and Netflix). Our reduction turns an arbitrary contextual bandit algorithm into a BIC one, with performance guarantees similar to those for the non-contextual version. From the selection bias perspective, this extension distinguishes the agents based on their observed attributes (\ie those included in the context), but protects against the selection bias that arises due to the \emph{unobserved} attributes.

Moreover, the reduction allows learning algorithms to incorporate arbitrary \emph{auxiliary feedback} that agents' actions may reveal. For example, a restaurant review may contain not only the overall evaluation of the agent's experience (\ie her reward), but also reveal her culinary preferences, which in turn may shed light on the popularity of other restaurants (\ie on the expected rewards of other actions). Further, an action can consist of multiple ``sub-actions", perhaps under common constraints, and the auxiliary feedback may reveal the reward for each sub-action. For instance, a detailed restaurant recommendation may include suggestions for each course, and a review may contain evaluations thereof. Such problems (without incentive constraints) have been actively studied in machine learning, under the names ``MAB with partial monitoring" and ``combinatorial semi-bandits"; see Section~\ref{sec:related-work} for relevant citations.

We allow for scenarios when the planner wishes to optimize his own utility which is misaligned with the agents'. Then rewards in our model still correspond to the agents' utilities, and planner's utility is observed by the algorithm as auxiliary feedback. For example, a vendor who recommends products to customers may favor more expensive products or products that are tied in with his other offerings. In a different setting, a planner may prefer \emph{less} expensive options: in a medical trial with substantial treatment costs, patients (who are getting these treatments for free) are only interested in their respective health outcomes, whereas a socially responsible planner may also factor in the treatment costs. Another example is a medical trial of several available immunizations for the same contagious disease, potentially offering different tradeoffs between the strength and duration of immunity and the severity of side effects. Hoping to free-ride on the immunity of others, a patient may assign a lower utility to a successful outcome than the government, and therefore prefer safer but less efficient options.

A black-box reduction such as ours is particularly desirable for the extended setting described above, essentially because it is not tied up to a particular variant of the problem. Indeed, contextual bandit algorithms in the literature heavily depend on the class of policies to optimize over, whereas our reduction does not. Likewise, algorithms for bandits with auxiliary feedback heavily depend on the particular kind of feedback.

\subsection{Our techniques and assumptions}

An essential challenge in BIC exploration is to incentivize agents to explore actions that appear suboptimal according to the agent's prior and/or the information currently available to the planner. The desirable incentives are created due to \emph{information asymmetry}: the planner knows more than the agents do, and the recommendation reveals a carefully calibrated amount of additional information. The agent's beliefs are then updated so that the recommended action now seems preferable to others, even though the algorithm may in fact be exploring in this round, and/or the prior mean reward of this action may be small.

Our problem is hopeless for some priors. For a simple example, consider a prior on two actions (whose expected rewards are denoted $\mu_1$ and $\mu_2$) such that $\E[\mu_1]>\E[\mu_2]$ and $\mu_1$ is statistically independent from $\mu_1-\mu_2$. Then, since no amount of samples from action $1$ has any bearing on $\mu_1-\mu_2$, a BIC algorithm cannot possibly incentivize agents to try action $2$. To rule out such pathological examples, we make some assumptions. Our detail-free result assumes that the prior is independent across actions, and additionally posits minor restrictions in terms of bounded rewards and full support. The black-box reduction posits a more abstract condition which allows for correlated priors, and includes independent priors as a special case (with similar minor restrictions). This condition is necessary for the case of two arms, if the algorithm is ``strongly BIC", \ie if each agent's utility is strictly maximized by the recommended action.

Our algorithms are based on (versions of) a common building block: an algorithm that incentivizes agents to explore at least once during a relatively short time interval (a ``phase"). The idea is to hide one round of exploration among many rounds of exploitation. An agent receiving a recommendation does not know whether this recommendation corresponds to exploration or to exploitation. However, the agents' Bayesian posterior favors the recommended action because the exploitation is much more likely. Information asymmetry arises because the agent cannot observe the previous rounds and the algorithm's randomness.

To obtain BIC algorithms with good performance, we overcome a number of technical challenges, some of which are listed below. First, an algorithm needs to convince an agent not to switch to several other actions: essentially, all actions with larger prior mean reward than the recommended action. In particular, the algorithm should accumulate sufficiently many samples of these actions beforehand. Second, we ensure that phase length --- \ie the sufficient size of exploitation pool --- does not need to grow over time. In particular, it helps not to reveal any
information to future agents (e.g., after each phase or at other ``checkpoints" throughout the algorithm). Third, for the black-box reduction we ensure that the choices of the bandit algorithm that we reduce from do not reveal any information about rewards. In particular, this consideration was essential in formulating the main assumption in the analysis (\propref{prop:general-persuasion} on page~\pageref{prop:general-persuasion}). Fourth, the detail-free algorithm cannot use Bayesian inference, and relies on sample average rewards to make conclusions about Bayesian posterior rewards, even though the latter is only an approximation for the former.

\OMIT{ 
\xhdr{Map of the technical content.}
We discuss technical preliminaries in Sections~\ref{sec:prelims}. The first technical result in the paper is a BIC algorithm for initial exploration in the special case of two arms (Section~\ref{sec:building-block}), the most lucid incarnation of the ``common building block" discussed above. Then we present the main results for BIC bandit exploration: the black-box reduction (Section~\ref{sec:black-box}) and the detail-free algorithm  with optimal ex-post regret (Section~\ref{sec:arms-elimination}). Then we proceed with a major extension to contexts and auxiliary feedback (Section~\ref{sec:general}). The proofs pertaining to the properties of the common prior are deferred to Section~\ref{sec:properties}. Conclusions and open questions are in Section~\ref{sec:conclusions}. The detail-free algorithm becomes substantially simpler for the special case of two actions. For better intuition, we provide a standalone exposition of this special case in Appendix~\ref{sec:DF-two-arms}.
} 

\subsection{Further discussion on medical trials}
\label{sec:medical}

We view patients' incentives as one of the major obstacles that inhibit medical trials in practice, or prevent some of them from happening altogether. This obstacle may be particularly damaging for large-scale trials that concern wide-spread medical conditions with relatively inexpensive treatments. Then finding suitable patients and providing them with appropriate treatments would be fairly realistic, but incentivizing patients to participate in sufficient numbers may be challenging. BIC exploration is thus a theoretical (and so far, highly idealized) attempt to mitigate this obstacle.

Medical trials have been one of the original motivations for studying MAB and exploration-exploitation tradeoff \citep{Thompson-1933,Gittins-index-79}. Bandit-like designs for medical trials belong to the realm of \emph{adaptive} medical trials
\citep{Chow-adaptive-2008}, which can also include other ``adaptive" features such as early stopping, sample size re-estimation, and changing the dosage.

``Message-restricted" algorithms (which recommend particular treatments to patients and reveal no other information) are appropriate for this domain. Revealing some (but not all) information about the medical trial is required to meet the standards of ``informed consent", as prescribed by various guidelines and regulations
\citep[see][for background]{CITI-GCP-2012}.
However, revealing information about clinical outcomes in an ongoing trial is currently not required, to the best of our understanding. In fact, revealing such information is seen as a significant threat to the statistical validity of the trial (because both patients and doctors may become biased in favor of better-performing treatments), and care is advised to prevent information leaks as much as possible
\citep[see ][for background and discussion, particularly pp. 26-30]{PCORI-adaptive-2012}.

Medical trials provide additional motivation for BIC bandit exploration with multiple actions. While traditional medical trials compare a new treatment against the placebo or a standard treatment, designs of medical trials with multiple treatments have been studied in the biostatistics literature \citep[e.g., see][]{Hellmich-multiarm-2001,Freidlin-multiarm-2008}, and are becoming increasingly important in practice \citep{Lancet-multiarm-2014,Redig-BasketTrials-2015}.
Note that even for the placebo or the standard treatment the expected reward is often not known in advance, as it may depend on the particular patient population.

BIC \emph{contextual} bandit exploration is particularly relevant to medical trials, as patients come with a lot of ``context" which can be used to adjust and personalize the recommended treatment. The context can include age, fitness levels, race or ethnicity, various aspects of the patient's medical history, as well as genetic markers (increasingly so as genetic sequencing is becoming more available). Context-dependent treatments (known as \emph{personalized medicine}) has been an important trend in the pharmaceutical industry in the past decade, especially genetic-marker-dependent treatments in oncology
\citep[e.g., see][]{Maitland-PersonalizedOncology-2011,Garimella-PersonalizedOncology-2015,Cell-PersonalizedOncology-2012}.
Medical trials for context-dependent treatments are more complex, as they must take the context into account. To reduce costs and address patient scarcity, a number of novel designs for such trials have been deployed
\citep[e.g., see][]{Maitland-PersonalizedOncology-2011,Garimella-PersonalizedOncology-2015}.
Some of the deployed designs are explicitly ``contextual", in that they seek the best policy --- mapping from patient's context to treatment \citep{Redig-BasketTrials-2015}. More advanced ``contextual" designs have been studied in biostatistics
\citep[e.g., see][]{Simon-trials-2005,Simon-trials-2007,Simon-trials-2010}.

\section{Related work}
\label{sec:related-work}

\xhdr{Exploration, exploitation, and incentives.}
There is a growing literature about a three-way interplay of exploration, exploitation, and incentives, comprising a variety of scenarios.

The study of mechanisms to incentivize exploration has been initiated by \citet{Kremer-JPE14}.
They mainly focus on deriving the Bayesian-optimal policy for the case of only two actions and deterministic rewards, and only obtain a preliminary result for stochastic rewards; a detailed comparison is provided below. \citet{Bimpikis-exploration-ms17}%
\footnote{\citet{Bimpikis-exploration-ms17} is concurrent and independent work with respect to the conference publication of this paper.}
consider a similar model with time-discounted rewards, focusing on the case of two arms. If expected rewards are known for one arm, they provide a BIC algorithm that achieves the ``first-best" utility. For the general case, they design an optimal BIC algorithm that is computationally inefficient, and propose a tractable heuristic based on the same techniques.
 Motivated by similar applications in the Internet economy, \citet{Che-13} propose a model with a continuous information flow and a continuum of consumers arriving to a recommendation system and derive a Bayesian-optimal incentive-compatible policy. Their model is technically different from ours, and is restricted to two arms and binary rewards. \citet{Frazier-ec14} consider a similar setting with monetary transfers, where the planner not only recommends an action to each agent, but also offers a payment for taking this action. In their setting, incentives are created via the offered payments rather than via information asymmetry.


Exploration-exploitation problems with self-interested agents have been studied in several other scenarios:
multiple agents engaging in exploration and benefiting from exploration performed by others, without a planner to coordinate them
    \citep[\eg][]{Bolton-econometrica99,Keller-econometrica05}
dynamic pricing with model uncertainty
    \citep[\eg][]{KleinbergL03,BZ09,BwK-focs13},
dynamic auctions
    \citep[\eg][]{AtheySegal-econometrica13,DynPivot-econometrica10,Kakade-pivot-or13},
pay-per-click ad auctions with unknown click probabilities
    \citep[\eg][]{MechMAB-ec09,DevanurK09,Transform-ec10-jacm},
as well as human computation
    \citep[\eg][]{RepeatedPA-ec14,Ghosh-itcs13,Krause-www13}.
In particular, a black-box reduction from an arbitrary MAB algorithm to an incentive-compatible algorithm is the main result in \cite{Transform-ec10-jacm}, in the setting of pay-per-click ad auctions with unknown click probabilities. The technical details in all this work (who are the agents, what are the actions, etc.) are very different from ours, so a direct comparison of results is uninformative.

\OMIT{ 
    \citet{Bolton-econometrica99} and \citet{Keller-econometrica05} consider settings in which agents can engage in exploration, and benefit from exploration performed by others. Unlike our setting, the agents are long-lived (present for many rounds), and there is no planner to incentivize efficient exploration.
} 


\xhdr{Subsequent work on incentivized exploration.} Several papers appeared subsequent to the conference publication of this paper. \citet{Bobby-Glen-ec16} consider costly information acquisition by self-interested agents (\eg real estate buyers or venture capitalists investigating a potential purchase/invenstment), and design a BIC mechanism to coordinate this process so as to improve social welfare. \citet{Bahar-ec16} extend the setting in \citet{Kremer-JPE14} --- BIC exploration with two actions and deterministic rewards --- to a known social network on the agents, where each agent can observe friend's recommendations in the previous rounds (but not their rewards). \citet{ICexplorationGames-ec16} extend BIC exploration to scenarios when multiple agents arrive in each round and their chosen actions may affect others' rewards. The paradigmatic motivating example is driving directions on \emph{Waze}, where drivers' choice of routes can create congestion for other drivers. The mechanism's recommendation is correlated across agents. Formally, it is a randomized mapping from history to joint actions which must be incentive-compatible in the sense of Bayesian correlated equilibrium.
\citet{Kempe-colt18,Jieming-multitypes18} investigated scenarios with heterogenous agents. Finally, \citet{kannan2017fairness} focuses on incentivizing \emph{fair} exploration, in the sense that worse arms are never (probabilistically) preferred to better ones.

\citet{ICexplorationGames-ec16,Jieming-multitypes18} take a different conceptual approach which is complementary to ours: rather than making assumptions to ensure that all agents are ``explorable" (\ie can be eventually chosen by an incentive-compatible algorithm), they start with no assumptions, and strive to explore all actions that can possibly be explored.

A closely related line of work \citep{bastani2017exploiting,kannan2018smoothed,Sven-aistats18,externalities-colt18}
considers the ``full revelation" scenario, when each agent sees the full history and chooses an action based on her own incentives. This corresponds to the ``greedy algorithm" in multi-armed bandits. While the greedy algorithm is known to suffer linear regret in the worst case, it performs well under some (fairly strong) assumptions on the diversity of agent population.

\xhdr{Information design.}
As mentioned in the Introduction, a single round of our model is as a version of the \emph{Bayesian Persuasion} game \citep{Kamenica-aer11} in which the planner's signal is the ``history" of the previous rounds. \citet{Rayo-jpe10} examine a related setting in which a planner incentivizes agents to make better choices via information disclosure. \citet{Ely-jpe15} and \citet{Horner-jpe15} consider information disclosure over time, in very different models: resp., releasing news over time to optimize suspense and surprise, and selling information over time. More background on the ``design of information structures" in theoretical economics can be found in \citep{BergemannMorris-infoDesign16,Taneva-infoDesign16}.

\xhdr{Detailed comparison to \citet*{Kremer-JPE14}.} 
While the expected reward is determined by the chosen action, we allow the realized reward to be stochastic. \citet{Kremer-JPE14} mainly focus on deterministic rewards, and only obtain a preliminary result for stochastic rewards. We improve over the latter result in several ways.

\citet{Kremer-JPE14} only consider the case of two actions, whereas we handle an arbitrary constant number of actions. Handling more than two actions is important in practice, because recommendation systems for products or experiences are rarely faced with only binary decisions. Further, medical trials with multiple treatments are important, too  \citep{Hellmich-multiarm-2001,Freidlin-multiarm-2008,Lancet-multiarm-2014}.
For multiple actions, convergence to the best action is a new and interesting result on its own, even regardless of the rate of convergence. Our extension to more than two actions is technically challenging, requiring several new ideas compared to the two-action case; especially so for the detail-free version.

We implement \emph{adaptive exploration}, when the exploration schedule is adapted to the previous observations, whereas the algorithm in \citet{Kremer-JPE14} is a BIC implementation of \emph{fixed exploration} --- a ``naive" MAB algorithm in which the exploration schedule is fixed in advance. This leads to a stark difference in ex-post regret. To describe these improvements, let us define the \emph{MAB instance} as a mapping from actions to their expected rewards, and let us be more explicit about the asymptotic ex-post regret rates as a function of the time horizon $T$ (i.e., the number of agents). \citet{Kremer-JPE14} only provides regret bound of $\tilde{O}(T^{2/3})$ for all MAB instances,%
\footnote{Here and elsewhere, the $\tilde{O}(\cdot)$ notation hides $\polylog(T)$ factors.}
whereas our algorithm achieves ex-post regret $\tilde{O}(\sqrt{T})$ for all MAB instances, and $\polylog(T)$ for MAB instances with constant ``gap" in the expected reward between the best and the second-best action. The literature on MAB considers this a significant improvement  (more on this below). In particular, the $\polylog(T)$ result is important in two ways:  it  quantifies the advantage of ``nice" MAB instances over the worst case, and of IID rewards over adversarially chosen rewards.%
\footnote{The MAB problem with adversarially chosen rewards only admits $\tilde{O}(\sqrt{T})$ ex-post regret in the worst case.}
The sub-par regret rate in \citet{Kremer-JPE14} is indeed a consequence of fixed exploration: such algorithms cannot have worst-case regret better than $\Omega(T^{2/3})$, and cannot achieve $\polylog(T)$ regret for MAB instances with constant gap \citep{MechMAB-ec09}.

In terms of information structure, the algorithm in \citet{Kremer-JPE14} requires all agents to have the same prior, and requires a very detailed knowledge of that prior; both are significant impediments in practice. Whereas our detail-free result allows the planner to have only a very limited knowledge of the prior, and allows the agents to have different priors.

Finally, \citet{Kremer-JPE14} does not provide an analog of our black-box reduction, and does not handle the various generalizations of BIC bandit exploration that the reduction supports.

\xhdr{Multi-armed bandits.} Multi-armed bandits (\emph{MAB}) have been studied extensively in Economics, Operations Research and Computer Science since \citep{Thompson-1933}. Motivations and applications range from medical trials to pricing and inventory optimization to driving directions to online advertising to human computation. A reader may refer to \citep{Bubeck-survey12} and \citep{Gittins-book11} for background on regret-minimizing and Bayesian formulations, respectively. Further background on related machine learning problems can be found in \citet{CesaBL-book}. Our results are primarily related to regret-minimizing MAB formulations with IID rewards \citep{Lai-Robbins-85,bandits-ucb1}.

Our detail-free algorithm builds on \emph{Hoeffding races} \citep{Races-nips93,Races-AIrev97}, a well-known technique in reinforcement learning. Its incarnation in the context of MAB is also known as \emph{active arms elimination} \citep{EvenDar-icml06}.

The general setting for the black-box reduction is closely related to three prominent directions in the work on MAB: contextual bandits, MAB with budgeted exploration, and MAB with partial monitoring. Contextual bandits have been introduced, under various names and models, in
\citep{Woodroofe79,bandits-exp3,Wang-sideMAB05,Langford-nips07}, and actively studied since then. We follow the formulation from \citep{Langford-nips07} and a long line of subsequent work \cite[\eg][]{policy_elim,monster-icml14}.
In MAB with budgeted exploration, algorithm's goal is to predict the best arm in a given number of rounds, and performance is measured via prediction quality rather than cumulative reward
\citep{EvenDar-colt02,Tsitsiklis-bandits-04,GuhaMunagala-stoc07,Null-soda09,Bubeck-alt09}.
In MAB with partial monitoring, auxiliary feedback is revealed in each round along with the reward. A historical account for this direction can be found in \citet{CesaBL-book}; see \citet{Bubeck-colt09,Bartok-MathOR14} for examples of more recent progress. Important special cases are MAB with graph-structured feedback
    \citep{Alon-nips13,Alon-colt15},
where choosing an arm would also bring feedback for adjacent arms,
and ``combinatorial semi-bandits"
    \citep{Gyorgy-jmlr07,Kale-nips10,Audibert-colt11,Wen-icml15},
where in each round the algorithm chooses a subset from some fixed ground set, and the reward for each chosen element of this set is revealed.

Improving the asymptotic regret rate from $\tilde{O}(T^\gamma)$, $\gamma>\tfrac12$ to $\tilde{O}(\sqrt{T})$ and from $\tilde{O}(\sqrt{T})$ to $\polylog(T)$ has been a dominant theme in the literature on regret-minimizing MAB (a through survey is beyond our scope, see \citet{Bubeck-survey12} for background and citations). In particular, the improvement from $\tilde{O}(T^{2/3})$ to $\tilde{O}(\sqrt{T})$ regret due to the distinction between fixed and adaptive exploration has been a major theme in
\citep{MechMAB-ec09,DevanurK09,DynPricing-ec12}.

As much of our motivation comes from human computation, we note that MAB-like problems have been studied in several other setting motivated by human computation. Most of this work has focused on crowdsourcing markets, addressing topics such as task-worker matching and pricing decisions; see \cite{Crowdsourcing-PositionPaper13} for a survey and discussion.


\xhdr{Fairness in ML.}
The connection between BIC exploration and fairness, observed in Section~\ref{sec:intro-scope}, relates this paper to a growing body of work on fairness in machine learning~\citep[\eg][]{dwork2012fairness,hardt2016equality,kleinberg2017inherent,chouldechova2017fair}. Some of this work studies fairness in the context of multi-armed bandits \citep[\eg][]{joseph2016fairness,liu2017calibrated,celis2017fair}, albeit with very different motivations and technical formulations. Most notably, \citep{joseph2016fairness,kearns2017meritocratic,kannan2017fairness} treat arms as applicants for limited resources such as bank loans, and require that worse applicants are not preferred to better ones.

\OMIT{
For example, such improvement was a main result in \cite{AuerOS/07,Bull-bandits14} in the context of continuum-armed bandits, in \cite{MechMAB-ec09,Transform-ec10-jacm} in the context of ``MAB mechanisms" for ad auctions,}

\OMIT{Our result discussed above (as well as the result in \citet{Kremer-JPE14}) requires the marginal priors --- the priors on the rewards on the individual actions --- to be mutually independent.}


\section{Model and preliminaries}
\label{sec:prelims}

\newcommand{\muPrior}{\mP^0}

We define the basic model, called \emph{BIC bandit exploration}; the generalizations are discussed in Section~\ref{sec:general}.  A sequence of $T$ agents arrive sequentially to the planner. In each round $t$, the interaction protocol is as follows: a new agent arrives, the planner sends this agent a signal $\sigma_t$, the agent chooses an action $i_t$ in a set $A$ of $m$ possible actions, receives a reward for this action, and leaves forever. The signal $\sigma_t$ includes a recommended action $I_t\in A$. This entire interaction is not observed by other agents. The planner knows the value of $T$. However, a coarse upper bound would suffice in most cases, with only a constant degradation of our results.


The planner chooses signals $\sigma_t$ using an algorithm, called \emph{recommendation algorithm}. If $\sigma_t=I_t=i_t$ (i.e., the signal is restricted to the recommended action, which is followed by the corresponding agent) then the setting reduces to multi-armed bandits (\emph{MAB}), and the recommendation algorithm is a bandit algorithm. To follow the MAB terminology, we will use \emph{arms} synonymously with \emph{actions}; we sometimes write ``play/pull an arm" rather than ``choose an action".

\xhdr{Rewards.}
For each arm $i$ there is a parametric family $\mD_i(\cdot)$ of reward distributions, parameterized by the expected reward $\mu_i$. The \emph{reward vector}
    $\vec{\mu}=(\mu_1,\ldots,\mu_m)$
is drawn from some prior distribution $\muPrior$. Conditional on the mean $\mu_i$, the realized reward when a given agent chooses action $i$ is drawn independently from distribution $\mD_i(\mu_i)$. In this paper we restrict attention to the case of single parameter families of distributions, however we do not believe this is a real restriction for our results to hold.

The prior $\muPrior$ and the tuple $\mD = (\mD_i(\cdot):\, i\in A)$ constitute the (full) Bayesian prior on rewards, denoted $\mP$. It is known to all agents and to the planner.%
\footnote{As mentioned in Introduction, we also consider an extension to a partially known prior.}
The expected rewards $\mu_i$ are not known to either.

For each arm $i$, let $\muPrior_i$ be the marginal of $\muPrior$ on this arm, and let $\mu_i^0=\E[\mu_i]$ be the prior mean reward. W.l.o.g., re-order the arms so that
    $\mu_1^0\geq \mu_2^0\geq \ldots\geq \mu_m^0$. The prior $\mP$ is independent across arms if the distribution $\muPrior$ is a product distribution, i.e., $\muPrior = \muPrior_1\times \ldots \times \muPrior_m$.


\xhdr{Incentive-compatibility.} Each agent $t$ maximizes her own Bayesian expected reward, conditional on any information that is available to him. Recall that the agent observes the planner's message $\sigma_t$, and does not observe anything about the previous rounds. Therefore, the agent simply chooses an action $i$ that maximizes the posterior mean reward $\E[\mu_i | \sigma_t]$. In particular, if the signal does not contain any new information about the reward vector $\mu$, then the agent simply chooses an action $i$ that maximizes the prior mean reward $\E[\mu_i]$.

We are interested in algorithms that respect incentives, in the sense that the recommended action $i=I_t$ maximizes the posterior mean reward. To state this property formally, let $\mE_{t-1}$ be the event that the agents have followed the algorithm's recommendations up to (and not including) round $t$.

\begin{definition}\label{def:BIC}
A recommendation algorithm is \emph{Bayesian incentive-compatible} (\emph{BIC}) if
\begin{align}\label{eqn:bic-constraint}
\E[\mu_i \mid \sigma_t,\, I_t=i, \mE_{t-1}]
    \geq \max_{j\in A} \E[\mu_j \mid \sigma_t,\, I_t=i, \mE_{t-1}]
    \qquad \forall t\in [T],\, \forall i\in A.
\end{align}
The algorithm is \emph{strongly BIC} if Inequality \eqref{eqn:bic-constraint} is always strict.
\end{definition}

\begin{remark}\label{rem:revelation-principle}
Throughout this paper, we focus on BIC recommendation algorithms with $\sigma_t=I_t$. As observed in \citet{Kremer-JPE14}, this restriction is w.l.o.g. in the following sense. First, any recommendation algorithm can be made BIC by re-defining $I_t$ to lie in
    $\argmax_i\,\E[\mu_i | \sigma_t]$.
Second, any BIC recommendation algorithm can be restricted to $\sigma_t=I_t$, preserving the BIC property. Note that the first step may require full knowledge of the prior and may be computationally expensive.
\end{remark}

Thus, for each agent the recommended action is at least as good as any other action. For simplicity, we assume the agents break ties in favor of the recommended action. Then the agents \emph{always} follow the recommended action.

\xhdr{Regret.}
The goal of the recommendation algorithm is to maximize the expected social welfare, i.e., the expected total reward of all agents. For BIC algorithms, this is just the total expected reward of the algorithm in the corresponding instance of MAB.

We measure algorithm's performance via the standard definitions of \emph{regret}.

\newcommand{\Rew}{\ensuremath{\mathtt{Rew}}}

\begin{definition}[Ex-post Regret]
The \emph{ex-post regret} of the algorithm is:
\begin{align}\label{eq:regret-defn}
R_\mu(T) = T(\max_i \mu_i) - \E\left[\sum_{t=1}^{T} \mu_{I_t} ~|~ \mu\right]
\end{align}

The \emph{Bayesian regret} of the algorithm is:
\begin{align}\label{eq:Bayesian-regret-defn}
R_\mP(T) = \E\left[T(\max_i \mu_i) - \sum_{t=1}^{T} \mu_{I_t}\right]
    = \E_{\mu\sim\muPrior}\; \left[R_\mu(T)\right].
\end{align}
\end{definition}

The ex-post regret is specific to a particular reward vector $\mu$; in the MAB terminology, it is sometimes called \emph{MAB instance}. In \eqref{eq:regret-defn} the expectation is taken over the randomness in the realized rewards and the algorithm, whereas in \eqref{eq:Bayesian-regret-defn} it is also taken over the prior. The last summand in \eqref{eq:regret-defn} is the expected reward of the algorithm.

Ex-post regret allows to capture ``nice" MAB instances. Formally, we consider  the \emph{gap} of a problem instance $\mu$, defined as
\begin{align}\label{def:gap}
\Delta = \mu^*-\max_{i: \mu_i<\mu^*} \mu_i,
\quad\text{where}\quad
\mu^* = \max_i \mu_i.
\end{align}
In words: the difference between the largest and the second largest expected reward. One way to characterize ``nice" MAB instances is via large gap. There are MAB algorithms with regret
    $O(\min(\tfrac{1}{\Delta} \log T, \sqrt{T\log T}))$
for a constant number of arms \cite{bandits-ucb1}.

The basic performance guarantee is expressed via Bayesian regret. Bounding the ex-post regret is particularly useful if the prior is not exactly known to the planner, or if the prior that everyone believes in may not quite the same as the true prior.
In general, a ex-post regret also guards against ``unlikely realizations''.
Besides, a bound on the ex-post regret is valid for every realization of the prior, which is reassuring for the social planner and also could take advantage of ``nice" realizations such as the ones with large ``gap".

\xhdr{Concentration inequalities.}
For our detail-free algorithm, we will use a well-known concentration inequality known as \emph{Chernoff-Hoeffding Bound}. This concentration inequality, in slightly different formulations, can be found in many textbooks and surveys \citep[e.g.,][]{MitzUpfal-book05}. We use a formulation from the original paper \citep{Hoeffding63}.


\begin{theorem}[Chernoff-Hoeffding Bound]
\label{thm:chernoff}
Consider $n$ I.I.D. random variables $X_1 \ldots X_n$ with values in $[0,1]$. Let
    $X = \tfrac{1}{n} \sum_{i=1}^n X_i$ be their average, and let $\mu = \E[X]$. Then:
\begin{align}\label{eq:chernoff}
\Pr{|X-\mu| < \delta } \geq 1- 2\, e^{-2 n\delta^2}
\qquad \forall \delta\in (0,1].
\end{align}
\end{theorem}

\noindent In a typical usage, we consider the high-probability event in \refeq{eq:chernoff} for a suitably chosen random variables $X_i$ (which we call the \emph{Chernoff event}), use the above theorem to argue that the failure probability is negligible, and proceed with the analysis conditional on the Chernoff event.

\xhdr{Conditional expectations.}
Throughout, we often use conditional expectations of the form $\E[A|B=b]$ and $\E[A|B]$ where $A$ and $B$ are random variables. To avoid confusion, let us emphasize that
    $\E[A|B=b]$
evaluates to a scalar, whereas $\E[A|B]$ is a random variable that maps values of $B$ to the corresponding conditional expectations of $A$. At a high level, we typically use $\E[A|B=b]$ in the algorithms' specifications, and we often consider $\E[A|B]$ in the analysis.

\section{Basic technique: sampling the inferior arm}
\label{sec:building-block}


A fundamental sub-problem in BIC bandit exploration is to incentivize agents to choose any arm $i\geq 2$ even once, as initially they would only choose arm $1$. (Recall that arms are ordered according to their prior mean rewards.) We provide a simple stand-alone BIC algorithm that samples each arm at least $k$ times, for a given $k$, and completes in time $k$ times a prior-dependent constant. This algorithm is the initial stage of the black-box reduction, and (in a detail-free extension) of the detail-free algorithm.

In this section we focus on the special case of two arms, so as to provide a lucid introduction to the techniques and approaches in this paper. (Extension to many arms, which requires some additional ideas, is postponed to Section \ref{sec:black-box}). We allow the common prior to be correlated across arms, under a mild assumption which we prove is necessary. For intuition, we explicitly work out the special cases of Gaussian priors and beta-binomial priors.

\subsection{Restricting the prior}

While consider the general case of correlated per-action priors, we need to restrict the common prior $\mP$ so as to give our algorithms a \emph{fighting chance}, because our problem is hopeless for some priors. As discussed in the Introduction, an easy example is a prior such that $\mu_1$ and $\mu_1-\mu_2$ are independent. Then samples from arm $1$ have no bearing on the conditional expectation of $\mu_1-\mu_2$, and therefore cannot possibly incentivize an agent to try arm $2$.

For a ``fighting chance", we assume that after seeing sufficiently many samples of arm $1$ there is a positive probability that arm $2$ is better. To state this property formally, we denote with $S_1^k$ the random variable that captures the first $k$ outcomes of arm $1$, and we let
\begin{align}\label{eq:defn-Xk}
     X^k := \E\left[ \mu_2-\mu_{1} \mid S_1^k \right]
\end{align}
be the conditional expectation of $\mu_2-\mu_1$ as a function $S^k_1$. We make the following assumption:

\begin{property}
\item\label{prop:persuasion}
$\Pr{X^k>0} >0$ for some prior-dependent constant $k=k_\mP<\infty$.
\end{property}

In fact, \propref{prop:persuasion} is ``almost necessary": it is necessary for a strongly BIC algorithm.

\begin{lemma}\label{lm:propA-app}
Consider an instance of BIC bandit exploration with two actions such that \propref{prop:persuasion} does not hold.  Then a strongly BIC algorithm never plays arm $2$.
\end{lemma}

\begin{proof}
Since \propref{prop:persuasion} does not hold, by Borel-Cantelli Lemma we have
\begin{align}\label{eq:app-propA-negation}
    \Pr{X^t\leq 0}=1 \quad\forall t\in \N.
\end{align}

We prove by induction on $t$ that the $t$-th agent cannot be recommended arm $2$. This is trivially true for $t=1$. Suppose the induction hypothesis is true for some $t$. Consider an execution of a strongly BIC algorithm. Such algorithm cannot recommend arm $2$ before round $t+1$. Then the decision whether to recommend arm $2$ in round $t+1$ is determined by the $t$ outcomes of action $1$. In other words, the event $U=\{I_{t+1}=2\}$ belongs to $\sigma(S^t)$, the sigma-algebra generated by $X^t$. Therefore,
\begin{align*}
\E[\mu_2-\mu_1|U]
    = \frac{1}{\Pr{U}} \int_U \mu_2-\mu_1
    = \E[X^t|U] \leq 0.
\end{align*}
The last inequality holds by \eqref{eq:app-propA-negation}. So a strongly BIC algorithm cannot recommend arm $2$ in round $t+1$. This completes the induction proof.
\end{proof}

\OMIT{ 
\begin{property}
\item\label{prop:persuasion}
There exist prior-dependent constants $k_\mP<\infty$ and $\tau_\mP, \rho_{\mP}>0$ such that
\begin{equation}
\Pr{X^k> \tau_\mP} \geq \rho_{\mP} \quad\forall k>k_\mP.
\end{equation}
\end{property}
} 

\subsection{Algorithm and analysis}

We provide a simple BIC algorithm that samples both arms at least $k$ times. The time is divided into $k+1$ phases of $L$ rounds each, except the first one, where $L$ is a parameter that depends on the prior. In the first phase the algorithm recommends arm~1 to $K=\max\{L,k\}$ agents, and picks the ``exploit arm" $a^*$ as the arm with a highest posterior mean conditional on these $K$ samples. In each subsequent phase, it picks one agent independently and uniformly at random from the $L$ agents in this phase, and use this agent to explore arm $2$. All other agents exploit using arm $a^*$. A more formal description is given in Algorithm \ref{alg:block}.

\begin{algorithm}[ht]
\SetKwInOut{Input}{Parameters}\SetKwInOut{Output}{Output}
\Input{$k,L\in\N$}
\BlankLine
\nl In the first $K=\max\{L,k\}$ rounds recommend arm $1$ \;
\nl Let $s_1^{K}=(r_1^1,\ldots,r_1^{K})$ be the corresp. tuple of rewards\;
\nl Let $a^* = \arg\max_{a\in \{1,2\}} \E[\mu_a| s_1^K]$, breaking ties in favor of arm $1$\;
\For{each phase $n = 2 \LDOTS k+1 $}{
\nl From the set $P$ of the next $L$ agents, pick one agent $p_0$ uniformly at random\;
\nl Every agent $p\in P-\{p_0\}$ is recommended arm $a^*$\;
\nl Player $p_0$ is recommended arm $2$
}
\caption{A BIC algorithm to collect $k$ samples of both arms.}
\label{alg:block}
\end{algorithm}

We prove that the algorithm is BIC as long as $L$ is larger than some prior-dependent constant. The idea is that, due to information asymmetry, an agent recommended arm $2$ does not know whether this is because of exploration or exploitation, but knows that the latter is much more likely. For large enough $L$, the expected gain from exploitation exceeds the expected loss from exploration, making the recommendation BIC.

\begin{lemma}\label{lem:block-ic}
Consider BIC bandit exploration with two arms. Assume \propref{prop:persuasion} holds for some $k_\mP$, and let $Y=X^{k_\mP}$.
Algorithm~\ref{alg:block} with parameters $k,L$ is BIC as long as
\begin{align}\label{set-L}
L \geq \max\left( k_\mP,\; 1+\frac{\mu_1^0-\mu_2^0}{\E[Y| Y> 0] \Pr{Y> 0}} \right).
\end{align}
The algorithm collects at least $k$ samples from each arm, and completes in $kL+\max(k,L)$ rounds.
\end{lemma}

Note that the denominator in \eqref{set-L} is strictly positive by \propref{prop:persuasion}. Indeed, the property implies that
    $\Pr{Y>\tau}=\delta>0$ for some $\tau>0$,
so $\E[Y| Y> 0]\geq \tau\delta>0$.

\begin{proof}
The BIC constraint is trivially satisfied in the initial phase. By a simple observation from \citet{Kremer-JPE14}, which also applies to our setting, it suffices to show that an agent is incentivized to follow a recommendation for arm $2$. Focus on one phase $n\geq 2$ of the algorithm, and fix some agent $p$ in this phase.
Let $Z=\mu_2-\mu_1$. It suffices to prove that
    $\E[Z| \rec{2}{p}]\; \Pr{\rec{2}{p}}\geq 0$,
i.e., that an agent recommended arm $2$ is not incentivized to switch to arm $1$.

Let $S_1^{K}$ be the random variable which represents the tuple of the initial samples $s_1^{K}$. Observe that the event $\rec{2}{p}$ is uniquely determined by the initial samples $S_1^K$ and by the random bits of the algorithm, which are independent of $Z$. Thus by the law of iterated expectations:
\begin{align*}
\E[Z | \rec{2}{p}]
    = \E\left[ \E\left[Z | S_1^K\right] | \rec{2}{p} \right]
    = \E[ X | \rec{2}{p}],
\end{align*}
where
    $X = X^K = \E\left[Z | S_1^K\right]$.
There are two possible disjoint events under which agent $p$ is recommended arm $2$: either
    $\Ev_1=\{X> 0\}$ or $\Ev_2=\{X\leq 0 \text{ and } p=p_0\}$.
Thus,
\begin{align*}
\E[Z | \rec{2}{p}] \Pr{\rec{2}{p}}
    &=\E[X| \rec{2}{p}]\Pr{\rec{2}{p}} \\
    &=\E[X|\Ev_1]\Pr{\Ev_1}+\E[X|\Ev_2]\Pr{\Ev_2},
\end{align*}
Observe that:
\[\Pr{\Ev_2} = \Pr{X \leq 0} \cdot \PrC{p=p_0}{X\leq 0} =  \tfrac{1}{L}\Pr{X\leq 0}.\]
Moreover, $X$ is independent of the event $\{p=p_0\}$. Therefore, we get:
\begin{align*}
\E[Z| \rec{2}{p}]\Pr{\rec{2}{p}}=~&\E[X| X>0]\Pr{X>0}+\tfrac{1}{L}\E[X| X\leq 0]\Pr{X\leq 0}.
\end{align*}
Thus, for
    $\E[Z| \rec{2}{p}]\geq 0$
it suffices to pick $L$ such that:
\begin{equation*}
L \geq \frac{-\E[X| X\leq 0 ]\Pr{X\leq 0}}{\E[X| X>0]\Pr{X>0}} = 1-\frac{\E[X]}{\E[X|X>0]\Pr{X>0}}.
\end{equation*}
The Lemma follows because
    $\E[X]=\E[Z]=\mu_2^0-\mu_1^0$
and
\begin{align*}
    \E[X|X>0]\,\Pr{X>0} \geq \E[Y|Y>0]\,\Pr{Y>0}.
\end{align*}
The latter inequality follows from Lemma~\ref{lm:propA-app-aux}, as stated and proved below. Essentially, it holds because $X=X^K$ conditions on more information than $Y=X^{k_\mP}$: respectively, the first $K$ samples and the first $k_\mP$ samples of arm $1$, where $K\geq k_\mP$.
\end{proof}

\begin{lemma}\label{lm:propA-app-aux}
Let $S$ and $T$ be two random variables such that  $T$ contains more information than $S$, i.e. if $\sigma(S)$ and $\sigma(T)$ are the sigma algebras generated by these random variables then $\sigma(S)\subset \sigma(T)$. Let $X^S=\E[f(\mu) | S]$ and $X^T=\E[f(\mu) | T]$ for some function $f(\cdot)$ on the vector of means $\mu$. Then:
\begin{equation}
\E[X^T | X^T>0] \Pr{X^T>0}\geq \E[X^S | X^S>0] \Pr{X^S>0}.
\end{equation}
\end{lemma}
\begin{proof}
Since the event $\{X^S>0\}$ is measurable in $\sigma(S)$  we have:
\begin{equation*}
\E[X^S | X^S>0] = \E\left[ \E[f(\mu) | S] | X^S >0\right] = \E[f(\mu) | X^S>0].
\end{equation*}
Since $\sigma(S)\subset \sigma(T)$:
\begin{equation*}
\E[f(\mu) | X^S>0]  = \E[ \E[f(\mu) | T] | X^S>0 ] = \E[ X^T | X^S>0].
\end{equation*}
Thus we can conclude:
\begin{align*}
\E[X^S | X^S>0]\, \Pr{X^S>0}
    &= \E[ X^T | X^S>0]\, \Pr{X^S>0} \\
    &= \int_{\{X^S>0\}} X^T\\
   &\leq \int_{\{X^T>0\}} X^T  = \E[X^T | X^T>0]\, \Pr{X^T>0}.
\end{align*}
\end{proof}

\subsection{Example: Gaussian priors}
\label{sec:ex:Gaussian}

To give more intuition on constants $L, k_\mP$ and random variable $Y=X^{k_\mP}$ from Lemma~\ref{lem:block-ic}, we present a concrete example where the prior on the expected reward of each arm is a Normal distribution, and the rewards are normally distributed conditional on their mean. We are interested in $X^k$, as defined in \eqref{eq:defn-Xk}; we derive an explicit formula for its distribution.  We will use $\mN(\mu,\sigma^2)$ to denote a normal distribution with mean $\mu$ and variance $\sigma^2$.

\begin{example}\label{ex:gaussians}
Assume the prior $\mP$ is independent across arms, $\mu_i \sim \mN(\mu_i^0,\sigma_i^2)$ for each arm $i$, and the respective reward $r_i^t$ is conditionally distributed as
    $\E[r_i^t \mid \mu_i] \sim \mN(\mu_i,\rho_i^2)$.
Then for each $k\in \N$
\begin{align}\label{eq:gaussians}
X^k \sim \mN\left( \mu_2^0-\mu_1^0,\;\frac{\sigma_1^2}{1+ \frac{\rho_1^2/\sigma_1^2}{k}}\right).
\end{align}
It follows that \propref{prop:persuasion} is satisfied for any $k=k_{\mP}\geq 1$.
\end{example}

\begin{proof}
Observe that
$r_1^t=\mu_1+\epsilon_1^t$, where $\epsilon_1^t\sim \mN(0,\rho_1^2)$ and $\mu_1=\mu_1^0+\zeta_1$, with $\zeta_1\sim \mN(0,\sigma_1^2)$.
Therefore,
\begin{align*}
\E[\mu_1 \mid S_1^k]
    &=\frac{\frac{1}{\sigma_1^2} \mu_1^0+\frac{1}{\rho_1^2}
        \sum_{t=1}^{k}r_1^t}{\frac{1}{\sigma_1^2}+k\cdot \frac{1}{\rho_1^2}}\\
    &= \mu_1^0 + \frac{k\cdot \frac{1}{\rho_1^2}}{\frac{1}{\sigma_1^2}+k\cdot
        \frac{1}{\rho_1^2}}\left(\zeta_1+ \frac{1}{k}\sum_{t=1}^{k}\epsilon_1^t\right) \\
    &= \mu_1^0 + \frac{k\cdot \sigma_1^2}{\rho_1^2+k\cdot \sigma_1^2 }\left(\zeta_1+
        \frac{1}{k}\sum_{t=1}^{k}\epsilon_1^t\right).
\end{align*}
Since
    $\zeta_1+\frac{1}{k}\sum_{t=1}^{k}\epsilon_1^t
        \sim \mN\left(0,\sigma_1^2 +\rho_1^2/k\right)$,
it follows that $ \E[\mu_1|S_1^k]$ is a linear transformation of a Normal distribution and therefore \begin{align*}
\E[\mu_1\mid S_1^k] \sim \mN\left( \mu_1^0,\; \left(\frac{k\cdot \sigma_1^2}{\rho_1^2+k\cdot \sigma_1^2 }\right)^2\cdot\left(\sigma_1^2 +\frac{\rho_1^2}{k}\right)\right) = \mN\left( \mu_1^0,\; \sigma_1^2 \frac{k\cdot \sigma_1^2}{\rho_1^2+k\cdot \sigma_1^2}\right).
\end{align*}
\refeq{eq:gaussians} follows because
    $X^k = \E[\mu_2-\mu_1|S_1^k] =\mu_2^0-\E[\mu_1|S_1^k]$.
\end{proof}

Observe that as $k\rightarrow \infty$, the variance of $X^k$ increases and converges to the variance of $\mu_1$. Intuitively, more samples can move the posterior mean further away from the prior mean $\mu_i^0$; in the limit of many samples $X^k$ will be distributed exactly the same as $\mu_2^0-\mu_1$.

While any value of the parameter $k_{\mP}$ is sufficient for BIC, one could choose $k_{\mP}$ to minimize the phase length $L$. Recall from Lemma~\ref{lem:block-ic} that the (smallest possible) value of $L$ is inversely proportional to the quantity
    $\E[Y|Y> 0]\Pr{Y> 0}$.
This is simply an integral over the positive region of a normally distributed random variable with negative mean, as given in \refeq{eq:gaussians}. Thus, increasing $k_{\mP}$ increases the variance of this random variable, and therefore increases the integral. For example, $k_{\mP}=\rho_1^2/\sigma_1^2$ is a reasonable choice if $\rho_1^2 \gg \sigma_1^2$.

\subsection{Example: beta-binomial priors}
\label{sec:ex:beta}

Our second example concerns beta-binomial priors, when the prior on the expected reward of each arm is a beta distribution, and the rewards are $0$ or $1$. We are interested in $X^k$, as defined in \eqref{eq:defn-Xk}. We use $\BetaD(\alpha,\beta)$ to denote the beta distribution with parameters $\alpha,\beta$; its expectation is $\tfrac{\alpha}{\alpha+\beta}$.

\begin{example}\label{ex:beta}
Assume there are only two arms, the prior $\mP$ is independent across arms, $\mu_i \sim \BetaD(\alpha_i,\beta_i)$ for each arm $i$, and the rewards are $0$ or $1$. Let
    $\mu_i^0 = \tfrac{\alpha_i}{\alpha_i+\beta_i}$
be the prior mean reward of arm $i$. Let $Z_k$ be the total reward of arm $1$ in the first $k$ samples. Then:

\begin{itemize}
\item[(a)] For each $\phi \in (0,\mu_2^0]$,
\[ \E[\mu_1 \mid S_1^k] <\phi
    \quad\Leftrightarrow\quad
    \frac{Z_k}{k} <\phi-\frac{G(\phi)}{k},
\]
where
    $G(\phi):= \alpha_1 - \phi(\alpha_1+\beta_1)
        = \alpha_1(1-\phi/\mu_1^0)>0$.

\item[(b)]
Take $\phi = \mu_2^0$ and $k>G(\phi)/\phi$. Then
    $\Pr{X^k>0}>0$.

\item[(c)] Take $\phi = \mu_2^0/2$ and $k>2\cdot G(\phi)/\phi$. Then
    $ \Pr{X^k\geq \phi} \geq \tfrac12\,\Pr{\mu_1\leq \phi/4} $.
Therefore,
\begin{align*}
\E[X^k \mid X^k>0]\, \Pr{X_k>0}
    &\geq \phi \cdot \Pr{X^k\geq \phi}\\
    &\geq \frac{\phi}{2}\cdot\Pr{\mu_1\leq \phi/4}.
\end{align*}
\end{itemize}

\end{example}

Thus, \propref{prop:persuasion} is satisfied whenever
    $k>G(\mu_2^0)/\mu_2^0$
according to part (b). Part (c) provides a more quantitative perspective on Lemma~\ref{lem:block-ic} (and also on \propref{prop:general-persuasion} in the next section).
\vspace{2mm}

\begin{proof}
We start with a well-known property of beta-binomial priors: after observing $k$ samples from arm $1$, the posterior distribution of $\mu_1$ is
    $\BetaD(\alpha_1+Z_k,\, \beta_1+k-Z_k)$.
Accordingly, the posterior mean reward of this arm is
    \[ \E[\mu_1\mid S_1^k] = \frac{\alpha_1+Z_k}{\alpha_1+\beta_1+k}. \]
Part (a) is obtained by a simple rearrangement of the terms.

Part (b) easily follows from Part (a):
\[
 \Pr{X^k>0} = \Pr{\E[\mu_1 \mid S_1^k] <\mu_2^0}
        = \Pr{Z_k/k<\eps}>0,
\]
where $\eps = \mu_2^0-G(\mu_2^0)/k>0$ (and the last inequality holds for any $\eps>0$).

Part (c) follows from part (a) as follows:
    $X^k\geq \phi$
holds whenever
    $Z_k/k\leq \phi/2$, so
\begin{align*}
 \Pr{X^k\geq \phi}
    &\geq \Pr{Z_k/k\leq \phi/2} \\
    &\geq \Pr{Z_k/k\leq \phi/2 \mid \mu_1\leq \phi/4}\cdot
    \Pr{\mu_1\leq \phi/4} \\
    &\geq \tfrac12\cdot \Pr{\mu_1\leq \phi/4}.
\end{align*}
The last inequality follows from Markov's inequality.
\end{proof}

%
%
%
%
%


\section{Black-box reduction from any algorithm}
\label{sec:black-box}

We present a ``black-box reduction" that uses a given bandit algorithm $\ALG$ as a black box to create a BIC algorithm $\ALGIC$ for BIC bandit exploration,  with only a minor increase in performance.

The reduction inputs two parameters $k,L$ and works as follows. We start with a ``sampling stage" which collects $k$ samples of each arm in only a constant number of rounds. The sampling stage for two arms is already implemented in Section~\ref{sec:building-block}. The sampling stage for $m>2$ arms is a non-trivial extension of the two-arms case, described below. Then we proceed with a ``simulation stage", divided in phases of $L$ rounds each. In each phase, we pick one round uniformly at random and ``dedicate" it to running the original algorithm $\ALG$. We run $\ALG$ in the ``dedicated" rounds only, ignoring the feedback from all other rounds. In the non-dedicated rounds, we recommend an ``exploit arm": an arm with the largest posterior mean reward conditional on the samples collected in the previous phases. A formal description is given in Algorithm \ref{alg:black-box}.

\begin{algorithm}[ht]
\SetKwInOut{Input}{Input}\SetKwInOut{Output}{Output}
\Input{A bandit algorithm $\ALG$; parameters $k,L\in \N$.}
\Input{A dataset $\mS^1$ which contains $k$ samples from each arm.}
\BlankLine
\nl Split all rounds into consecutive phases of $L$ rounds each\;
\For{each phase $n = 1,\ldots$}{
\nl Let $a^*_n = \arg\max_{i\in A} \E[\mu_i | \mS^n]$ be the ``exploit arm"\;
\nl Query algorithm $\ALG$ for its next arm selection $i_n$\;
\nl Pick one agent $p$ from the $L$ agents in the phase uniformly at random\;
\nl Every agent in the phase is recommended $a^*_n$, except agent $p$ who is recommended $i_n$\;
\nl Return to algorithm $\ALG$ the reward of agent $p$\;
\nl Set $\mS^{n+1} = \mS^n \cup \{ \text{all samples collected in this phase} \}$
}
\caption{Black-box reduction: the simulation stage}
\label{alg:black-box}
\end{algorithm}

The sampling stage for $m>2$ arms works as follows. In order to convince an agent to pull an arm $i>1$, we need to convince him not to switch to any other arm $j<i$. We implement a ``sequential contest'' among the arms. We start with collecting samples of arm $1$. The rest of the process is divided into $m-1$ phases $2,3 \LDOTS m$, where the goal of each phase $i$ is to collect the samples of arm $i$. We maintain the ``exploit arm" $a^*$, the arm with the best posterior mean given the previously collected samples of arms $j<i$. The $i$-th phase collects the samples of arm $i$ using a procedure similar to that in Section~\ref{sec:building-block}: $k$ agents chosen u.a.r. are recommended arm $i$, and the rest are recommended the exploit arm. Then the current exploit arm is compared against arm $i$, and the ``winner" of this contest becomes the new exploit arm. The pseudocode is in Algorithm~\ref{alg:block-m}.

\begin{algorithm}[ht]
\SetKwInOut{Input}{Parameters}\SetKwInOut{Output}{Output}
\Input{$k,L\in\N$; the number of arms $m$.}
\BlankLine
\nl For the first $k$ agents recommend arm $1$, let $s_1^k=\{r_1^1,\ldots,r_1^k\}$ be the tuple of rewards\;
\For{each arm $i>1$ in increasing lexicographic order}
{
    \nl Let $a^* = \arg\max_{a\in A} \E[\mu_a| s_1^k \LDOTS s_{i-1}^k]$,
        breaking ties lexicographically\;
	\nl From the set $P$ of the next $L\cdot k$ agents, pick a set $Q$ of $k$ agents uniformly at random\;
	\nl Every agent $p\in P-Q$ is recommended arm $a^*$\;
	\nl Every agent $p\in Q$ is recommended arm $i$\;
	\nl Let $s_i^k$ be the tuple of rewards of arm $i$ returned by agents in $Q$
}
\caption{Black-box reduction: the sampling stage.}
\label{alg:block-m}
\end{algorithm}

\newpage
\xhdr{Incentive-compatibility guarantees.}
We make an assumption on the prior $\mP$:

\begin{property}
\item\label{prop:general-persuasion}
Let $S_i^k$ be the the random variable representing $k$ rewards of action $i$,
and let
\begin{align}\label{eq:prop:general-persuasion}
X_{i}^k = \min_{j\neq i} \E[\mu_i-\mu_{j} | S_1^k,\ldots, S_{i-1}^k].
\end{align}
There exist prior-dependent constants $k_\mP, \tau_\mP, \rho_\mP>0$ such that
\begin{align*}
\Pr{X_{i}^k> \tau_\mP} \geq \rho_{\mP} \qquad\forall k>k_\mP,\, i\in A.
\end{align*}
\end{property}

\noindent Informally: any given arm $i$ can ``a posteriori'' be the best arm by margin $\tau_\mP$ with probability at least $\rho_\mP$ after seeing sufficiently many samples of each arm $j<i$. We deliberately state this property in such an abstract manner in order to make our BIC guarantees as inclusive as possible, e.g., allow for correlated priors, and avoid diluting our main message by technicalities related to convergence properties of Bayesian estimators.

\begin{remark}
For the special case of two arms, \propref{prop:general-persuasion} follows from \propref{prop:persuasion}, see Section~\ref{sec:properties-more} for a proof. Recall that the latter property is necessary for any strongly BIC algorithm.
\end{remark}

\begin{remark}
For quantitative perspective, consider the special case of two arms and beta-binomial priors. See Example~\ref{ex:beta}(c) for an instantiation of
    \propref{prop:general-persuasion}
with specific parameters.
\end{remark}

\begin{theorem}[BIC]
\label{thm:bb-ic}
Assume \propref{prop:general-persuasion} holds with constants $k_\mP,\tau_\mP,\rho_\mP$. Then the black-box reduction is BIC, for any bandit algorithm $\ALG$, if the parameters $k,L$ are larger than some prior-dependent constant. More precisely, it suffices to take
\[ k\geq k_\mP
   \text{ and }
    L\geq 2+\frac{\mu_{\max}^0-\mu_m^0}{\tau_\mP\cdot \rho_\mP},
   \text{ where }
   \mu_{\max}^0=\E\left[ \max_{i\in A}\mu_i\right]. \]
\end{theorem}

We prove Theorem~\ref{thm:bb-ic} in Sections~\ref{sec:bb-BIC-simulation} and ~\ref{sec:bb-BIC-sampling}, separately for the two stages of the algorithm.

If the common prior $\mP$ puts positive probability on the event that
    $\{\mu_i-\max_{j\neq i}\mu_j\geq \tau\}$,
then \propref{prop:general-persuasion} is plausible because the conditional expectations tend to approach the true means as the number of samples grows. We elucidate this intuition via a concrete, and yet fairly general, example, where we focus on priors that are independent across arms.
(See Section~\ref{sec:properties} for the proof).

\begin{lemma}\label{lm:sufficient-conditions-on-P}
\propref{prop:general-persuasion} holds for a constant number of arms under the following assumptions:
\begin{OneLiners}
\item[(i)] The prior $\mP$ is independent across arms.

\item[(ii)] The prior on each $\mu_i$ has full support on some bounded interval $[a,b]$.~%
\footnote{That is, the prior on $\mu_i$ assigns a strictly positive probability to every open subinterval of $[a,b]$, for each arm $i$.}

\item[(iii)] The prior mean rewards are all distinct:
    $\mu_1^0>\mu_2^0>\ldots > \mu_m^0$.
\item[(iv)] The realized rewards are bounded (perhaps on a larger interval).
\end{OneLiners}
\end{lemma}

\OMIT{\item The parameterized reward distribution $\mD_i(\mu_i)$ of each arm $i$ belongs to the exponential family of distributions. Its Fisher Information relative to $\mu_i$ is al least $\phi>0$,
    for all $\mu_i\in [a,b]$.}

\OMIT{ 
\begin{remark}\label{rem:sufficient-conditions-on-P}
Reward distributions $\mD_i(\mu_i)$ which satisfy the above conditions include the normal distribution, Bernoulli distribution,  \ascomment{what else?}
See Section~\ref{sec:properties} the proof of Lemma~\ref{lm:sufficient-conditions-on-P} and additional details.
\end{remark}
} 



\xhdr{Performance guarantees.}
The algorithm's performance can be quantified in several ways: in terms of the total expected reward, the average expected reward, and in terms of Bayesian regret; we state guarantees for all three. Unlike many results on regret minimization, we do \emph{not} assume that expected rewards are bounded.

\begin{theorem}[Performance]\label{thm:reduction-performance}
Consider the black-box reduction with parameters $k,L$, applied to bandit algorithm $\ALG$. Let $c=k+Lk$ be the number of rounds in the sampling stage. Then:

\begin{itemize}
\item[(a)]
Let $\widehat{U}_{\ALG}(\tau_1,\tau_2)$ be the average Bayesian-expected reward of algorithm $\ALG$ in rounds
$[\tau_1,\tau_2]$. Then
\begin{align*}
    \widehat{U}_{\ALGIC}(c+1,c+\tau) \geq  \widehat{U}_{\ALG}(1,\flr{\tau/L})
    \qquad\text{for any duration $\tau$}.
\end{align*}

\item[(b)] Let $U_{\ALG}(t)$ be the total Bayesian-expected reward of algorithm $\ALG$ in the first $\flr{t}$ rounds. Then
\begin{align*}
U_{\ALGIC}(t) \geq  L\cdot U_{\ALG}\left(\frac{t-c}{L}\right)
    + c\cdot \E\left[ \min_{i\in A} \mu_i \right].
\end{align*}

\item[(c)] Let $R_{\ALG}(t)$ be the Bayesian regret of algorithm $\ALG$ in the first $\flr{t}$ rounds.
\begin{align*}
R_{\ALGIC}(t) \leq L\cdot R_{\ALG}(t/L)
        + c\cdot \E\left[ \max_{i\in A} \mu_i - \min_{i\in A} \mu_i \right].
\end{align*}
In particular, if $\mu_i \in [0,1]$ for all arms $i$, and the original algorithm $\ALG$ achieves asymptotically optimal Bayesian regret $O(\sqrt{t})$, then so does $\ALGIC$.

\end{itemize}
\end{theorem}

\begin{proof}
Recall that $a^*_n$ and $i_n$ are, resp., the ``exploit" arm and the original algorithm's selection in phase $n$. Part (a) follows because
    $\E\left[ \sum_{n=1}^{N} {\mu_{i_n}} \right]  =\widehat{U}_\ALG(1,N)$,
where $N$ is the number of phases, and by definition of the ``exploit" arm we have
    $\E[\mu_{a^*_n}|\mS^n]\geq \E[\mu_{i_n}|\mS^n]$
for each phase $n$.

Parts (b) and (c) follow easily from part (a) by going, resp., from average rewards to total rewards, and from total rewards to regret. We have no guarantee whatsoever on the actions chosen in the first $c$ rounds, so we assume the worst: in part (b) we use the worst possible Bayesian-expected reward
     $\E\left[ \max_{i\in A} \mu_i \right]$,
and in part (c) we use worst possible per-round Bayesian regret
    $\E\left[ \max_{i\in A} \mu_i - \min_{i\in A} \mu_i \right]$.
\end{proof}

\subsection{The black-box reduction with predictions}
\label{sec:bb-predictions}

We can also allow the original algorithm $\ALG$ to output a \emph{prediction} in each round. For our purposes, the prediction is completely abstract: it can be any element in the given set of possible predictions. It can be a predicted best arm, the estimated gap between the best and the second-best arms (in terms of their expected rewards), or any other information that the planner wishes to learn. For example, as pointed out in the Introduction, the planner may wish to learn the worst action (or several worst actions, and possibly along with their respective rewards), in order to eliminate these actions via regulation and legislation.

The black-box reduction can be easily modified to output predictions, visible to the planner but not to agents. In each phase $n$, the reduction records the prediction made by $\ALG$, denote it by $\phi_n$. In each round of the initial sampling stage, and in each round of the first phase ($n=1$) the reduction outputs the ``null" prediction. In each round of every subsequent phase $n\geq 2$, the reduction outputs $\phi_{n-1}$.

Predictions are not visible to agents and therefore have no bearing on their incentives; therefore, the resulting algorithm is BIC by Theorem~\ref{thm:bb-ic}.

We are interested in the \emph{sample complexity} of $\ALGIC$: essentially, how the prediction quality grows over time. Such guarantee is meaningful even if the performance guarantees from Theorem~\ref{thm:reduction-performance} are not strong. We provide a very general guarantee: essentially, $\ALGIC$ learns as fast as $\ALG$, up to a constant factor.

\begin{theorem}
Consider the black-box reduction with parameters $k,L$, and let $c=k+Lk$ be the number of rounds in the sampling stage. Then the prediction returned by algorithm $\ALGIC$ in each round $t>c+L$ has the same distribution as that returned by algorithm $\ALG$ in round $\flr{(t-c)/L}$.
\end{theorem}

\subsection{Proof of Theorem~\ref{thm:bb-ic}: the sampling stage}
\label{sec:bb-BIC-sampling}

To guarantee the BIC property for the sampling stage of the algorithm, it suffices to have a somewhat weaker bound on parameter $L$ than stated in Theorem~\ref{thm:bb-ic}. Namely, we will only assume
    \[ L \geq 1+\frac{\mu_1^0-\mu_m^0}{\tau_\mP \cdot \rho_\mP}.\]

The algorithm can be split into $m$ phases, where each phase except the first one lasts $L\cdot k$ rounds, and agents are aware which phase they belong to. We will argue that each agent $p$ in each phase $i$ has no incentive not to follow the recommended arm. If an agent in phase $i$ is recommended any arm $j\neq i$, then she knows that this is because this is the exploit arm $a^*_i$; so by definition of the exploit arm it is incentive-compatible for the agent to follow it. Thus it remains to argue that an agent $p$ in phase $i$ who is recommended arm $i$ does not want to deviate to some other arm $j\neq i$, i.e., that
$\E[ \mu_i - \mu_j | \rec{i}{p} ] \Pr{\rec{i}{p}} \geq 0$.

Denote
    $X_{ij} := \E[\mu_i-\mu_j | S_1^k, \ldots, S_{i-1}^k]$
and
    $X_i := X_i^k =\min_{j\neq i} X_{ij}$.

By the law of iterated expectations:
\begin{align*}
\E[ \mu_i - \mu_j | \rec{i}{p} ] \Pr{\rec{i}{p}}
    &= \E[ \E[\mu_i-\mu_j | S_1^k,\ldots, S_{i-1}^k] | \rec{i}{p} ] \;\Pr{\rec{i}{p}} \\
    &= \E[ X_{ij} | \rec{i}{p}] \Pr{\rec{i}{p}}.
\end{align*}

There are two possible disjoint events under which agent $p$ is recommended arm $i$, either $\Ev_1=\{X_i> 0\}$ or $\Ev_2=\{X_i\leq 0 \text{ and } p\in Q\}$.  Thus:
\begin{align*}
\E[\mu_i-\mu_j | \rec{i}{p}] \Pr{\rec{i}{p}}
    &=\E[X_{ij}| \rec{i}{p}]\Pr{\rec{i}{p}} \\
    &= \E[X_{ij}|\Ev_1]\Pr{\Ev_1}+\E[X_{ij}|\Ev_2]\Pr{\Ev_2}.
\end{align*}
Observe that:
\[\Pr{\Ev_2} = \Pr{X_i \leq 0} \cdot \PrC{p\in Q}{X_i\leq 0} =  \frac{1}{L}\Pr{X_i\leq 0}.\]
Moreover, $X_{ij}$ is independent of the event $\{p\in Q\}$. Therefore, we get:
\begin{align*}
\E[\mu_i-\mu_j| \rec{i}{p}]\Pr{\rec{i}{p}}=~&\E[X_{ij}| X_{i}>0]\Pr{X_{i}>0}+\tfrac{1}{L}\E[X_{ij}| X_i\leq 0]\Pr{X_i\leq 0}\\
=~& \E[X_{ij}| X_{i}>0]\Pr{X_{i}>0} \left(1-\tfrac{1}{L}\right) + \tfrac{1}{L} \E[X_{ij}]\\
=~& \E[X_{ij}| X_{i}>0]\Pr{X_{i}>0} \left(1-\tfrac{1}{L}\right) + \tfrac{1}{L} \left(\mu_i^0-\mu_j^0\right)\\
\geq~& \E[X_{i}| X_{i}>0]\Pr{X_{i}>0} \left(1-\tfrac{1}{L}\right) + \tfrac{1}{L} \left(\mu_i^0-\mu_j^0\right).
\end{align*}
In the above, we used the facts that
    $\E[X_{ij}] = \E[ \mu_i - \mu_j] = \mu_i^0-\mu_j^0$ and that $X_{ij}\geq X_i$.
Thus it suffices to pick $L$ such that:
\begin{align*}
L \geq  1+\frac{\mu_j^0-\mu_i^0}{\E[X_{i}|X_i>0]\Pr{X_i>0}}.
\end{align*}
The theorem follows by our choice of $L$ and since $\mu_j^0-\mu_i^0\leq \mu_1^0-\mu_i^0$ for all $j\in A$, and by \propref{prop:general-persuasion} we have that $\E[X_{i}|X_i>0]\Pr{X_i>0} \geq \tau_\mP \cdot \rho_\mP$.

\subsection{Proof of Theorem~\ref{thm:bb-ic}: the simulation stage}
\label{sec:bb-BIC-simulation}

Instead using \propref{prop:general-persuasion} directly, we will work with a somewhat weaker corollary thereof:
\begin{property}
\item \label{prop:switch}
Let $\Lambda_k$ be a random variable representing $k_i\geq k$ samples from each action $i$.
There exist prior-dependent constants $k^*_\mP<\infty$ and $\tau_\mP, \rho_\mP>0$ such that
\begin{align*}
\forall k\geq k^*_\mP \quad \forall i\in A\qquad
\Pr{\min_{j\in A-\{i\}} \E\left[\mu_i-\mu_{j}|\Lambda_k \right] \geq \tau_\mP} \geq \rho_\mP.
\end{align*}
\end{property}
\propref{prop:switch} is derived in Section~\ref{sec:properties-more}.

Let us consider phase $n$ of the algorithm. We will argue that any agent $p$ who is recommended some arm $j$ does not want to switch to some other arm $i$. More formally,
we assert that
\begin{align*}
    \E[ \mu_j-\mu_i | \rec{j}{p} ] \; \Pr{\rec{j}{p}}\geq 0.
\end{align*}

Let $S^n$ be the random variable which represents the dataset $\mS^n$ collected by the algorithm by the beginning of phase $n$. Let
    $X_{ji}^n = \E[\mu_j-\mu_i| S^n]$
and
    $X_j^n=\min_{i\in A/\{j\}}X_{ji}^n$.
It suffices to show:
\begin{align*}
\E\left[ X_{j}^n| \rec{j}{p}\right] \;\Pr{\rec{j}{p}}\geq 0\;.
\end{align*}

Call an agent \emph{unlucky} if it is chosen as agent $p$ in Algorithm~\ref{alg:black-box}. There are two possible ways that agent $p$ is recommended arm $j$: either arm $j$ is the best posterior arm, i.e., $X_j^n > 0$, and agent $i$ is not the unlucky agent or $X_j^n< 0$, in which case arm $j$ cannot possibly be the exploit arm $a^*_n$ and therefore, $\rec{j}{p}$ happens only if agent $p$ was the unlucky one agent among the $L$ agents of the phase to be recommended the arm of algorithm $\ALG$ and algorithm $\ALG$ recommended arm $j$. There are also the followings two cases: (i) $X_j^n=0$ or (ii) $X_j^n>0$, agent $p$ is the unlucky agent and algorithm $\ALG$ recommends arm $j$. However, since we only want to lower bound the expectation and since under both these events $X_{j}^n\geq 0$, we will disregard these events, which can only contribute positively.

Denote with $U$ the index of the ``unlucky" agent selected in phase $n$ to be given recommendation of the original algorithm $\ALG$. Denote with $\ALG_j^n$ the event that algorithm $\ALG$ pulls arm $j$ at iteration $n$. Thus by our reasoning in the previous paragraph:
\begin{align}
\E\left[ X_{j}^n| \rec{j}{p} \right] \; \Pr{\rec{j}{p}}
\geq& \E[X_{j}^n| {X_j^n\geq 0}, U\neq p] \; \Pr{{X_j^n\geq 0}, U\neq p} \nonumber \\
&+\E[X_{j}^n| {X_j^n < 0}, U=p, \ALG_j^n]\;\Pr{{X_j^n < 0}, U=p, \ALG_j^n}
    \label{eq:pf-bb-ic:1}
\end{align}
Since $X_{j}^n$ is independent of $U$, we can drop conditioning on the event $\{U=p\}$ from the conditional expectations in \refeq{eq:pf-bb-ic:1}. Further, observe that:
\begin{align*}
\Pr{{X_j^n < 0}, U=p, \ALG_j^n}
    &= \PrC{U=p}{{X_j^n < 0},  \ALG_j^n} \Pr{{X_j^n < 0},  \ALG_j^n} \\
    &= \tfrac{1}{L} \Pr{{X_j^n < 0},  \ALG_j^n} \\
\Pr{{X_j^n\geq 0}, U\neq p}
    &= \PrC{U\neq p}{{X_j^n\geq 0}} \Pr{{X_j^n\geq 0}} \\
    &= \left(1-\tfrac{1}{L}\right) \Pr{{X_j^n\geq 0}}
\end{align*}
Plugging this into \refeq{eq:pf-bb-ic:1}, we can re-write it as follows:
\begin{align}
\E\left[ X_{j}^n| \rec{j}{p} \right]\,\Pr{\rec{j}{p}}
    &\geq \left(1-\tfrac{1}{L}\right)\; \E[X_{j}^n| {X_j^n\geq 0}]\; \Pr{{X_j^n\geq0}} \nonumber \\
    &    \qquad+\tfrac{1}{L}\E[X_{j}^n| {X_j^n < 0},\ALG_j^n]\Pr{{X_j^n < 0},\ALG_j^n}.
        \label{eq:pf-bb-ic:2}
\end{align}
Observe that conditional on $\{X_j^n<0\}$, the integrand of the second part in the latter summation is always non-positive. Thus integrating over a bigger set of events can only decrease it, i.e.:
\begin{align*}
\E[X_j^n | X_j^n< 0, \ALG_j^n]\Pr{X_j^n< 0,\ALG_j^n} \geq~&  \E[X_j^n | X_j^n< 0]\Pr{X_j^n< 0}
\end{align*}
Plugging this into \refeq{eq:pf-bb-ic:2}, we can lower bound the expected gain as:
\begin{align*}
\E\left[ X_j^n| \rec{j}{p}\right]\;\Pr{\rec{j}{p}}
    \geq~& \left(1-\tfrac{1}{L}\right)\E[X_j^n| X_j^n\geq  0]\Pr{X_j^n\geq 0}
        +\tfrac{1}{L}\E[X_j^n | X_j^n< 0]\Pr{X_j^n< 0}\\
    =~&\E[X_j^n| X_j^n\geq 0]\Pr{X_j^n\geq 0} \tfrac{L-2}{L} + \tfrac{1}{L} \E[X_j^n].
\end{align*}
Now observe that:
\begin{align*}
\E[X_j^n]
    &= \E\left[\min_{i\in A/\{j\}}\E[\mu_j-\mu_i| S^n]\right]
    \geq \E\left[\E\left[\min_{i\in A/\{j\}} (\mu_j - \mu_i) | S^n \right]\right]
    = \E\left[\mu_j - \max_{i\in A/\{j\}} \mu_i\right] \\
    &= \mu_j^0 - \mu_{\max}^0.
\end{align*}
By property \refprop{prop:switch}, we have that $\Pr{X_j^n\geq \tau_\mP}\geq \rho_\mP$. Therefore:
\begin{align*}
\E[X_j^n| X_j^n\geq 0]\Pr{X_j^n \geq 0} \geq \tau_\mP \cdot \rho_\mP >0.
\end{align*}
Thus the expected gain is lower bounded by:
\begin{align*}
\E\left[X_j^n| \rec{j}{p}\right]\,\Pr{\rec{j}{p}}
    \geq~ \tau_\mP \cdot \rho_\mP \frac{L-2}{L} + \frac{1}{L} (\mu_j^0-\mu_{\max}^0).
\end{align*}
The latter is non-negative if:
\begin{align*}
L \geq \frac{\mu_{\max}^0-\mu_j^0}{\tau_\mP\cdot \rho_\mP}+2
\end{align*}

Since this must hold for every any $j$, we get that for incentive compatibility it suffices to pick:
\begin{align*}
L \geq \frac{\max_{j\in A} (\mu_{\max}^0-\mu_j^0)}{\tau_\mP\cdot \rho_\mP}+2 = \frac{\mu_{\max}^0-\mu_m^0}{\tau_\mP\cdot \rho_\mP}+2. \qquad\qquad\qedhere
\end{align*}


\section{Ex-post regret bounds with a detail-free algorithm}
\label{sec:arms-elimination}

We present a BIC algorithm with near-optimal bounds on ex-post regret, both for ``nice" MAB instances and in the worst case over all MAB instances. Our algorithm, called \DFAlg, is ``detail-free", as it requires only a limited knowledge of the prior. This is an improvement over the results in the preceding section, whcih relied on the precise knowledge of the prior, and only provided performance guarantees in expectation over the prior.

\subsection{Results and discussion}

We use an assumption to restrict the common prior, like in the previous sections. We need this assumption not only to give the algorithm a ``fighting chance", but also to ensure that the sample average of each arm is a good proxy for the respective posterior mean. Hence, we focus on priors that are independent across arms, assume bounded rewards, and full support for each $\mu_i$.

\begin{property}
\item\label{prop:DF}
The prior is independent across arms; rewards lie in $[0,1]$;
each $\mu_i$ has full support on $[0,1]$.
\end{property}

\noindent However, we make no assumptions on the distribution of each $\mu_i$ and on the parameterized reward distributions $\mD(\mu_i)$, other than the bounded support.

Our algorithm needs to know the ordering of the arms, an approximation to the smallest prior mean reward, and an upper bound on a certain prior-independent threshold.

\begin{assumption}\label{ass:detail-free}
Algorithm \DFAlg knows the following:
\begin{itemize}

\item the ordering of the arms according to their prior mean rewards;

\item the smallest prior mean reward $\mu_m^0$, or a parameter $\hat{\mu}$ such that
         $\mu_m^0 \leq \hat{\mu} \leq  2\,\mu_m^0$;

\item a certain threshold $N_\mP$ that depends only on the prior $\mP$, or a parameter $N\geq N_\mP$.

\end{itemize}
\end{assumption}

\noindent A loose upper bound on the threshold $N_\mP$ suffices to guarantee BIC property, but hurts algorithm's regret. Since the algorithm does not know the exact prior, it avoids posterior updates, and instead estimates the expected reward of each arm with its sample average.

The high-level guarantees for \DFAlg can be stated as follows:

\begin{theorem}\label{thm:DF}
Algorithm \DFAlg is parameterized by two positive numbers $(\hat{\mu},N)$.

\begin{itemize}
\item[(a)] The algorithm is BIC if \propref{prop:DF} holds, $\mu_m^0>0$, and Assumption~\ref{ass:detail-free} is satisfied.

\item[(b)]
The algorithm achieves ex-post regret
\begin{align*}
R(t)
    \leq \min\left( C_N\cdot\Delta + \frac{O(m\log T N)}{\Delta},\; t \Delta \right)
    \leq C_N + O\left( \sqrt{mt \cdot \log(TN)} \right).
\end{align*}
Here $\Delta$ is the gap, defined as per \eqref{def:gap}, $T$ is the time horizon, $m$ is the number of arms, and $C_N = N+N^2(m-1)$. This regret bound holds for each round $t\leq T$ and any value of the parameters.
\end{itemize}
\end{theorem}

The dependence on the prior $\mP$ creeps in via threshsold $\N_\mP$ and parameter $N\geq \N_\mP$.

Due to the detail-free property we can allow agents to have different priors (subject to some assumptions). For example, all agents can start with the same prior, and receive idiosyncratic signals before the algorithm begins. Further, agents do not even need to know their own priors exactly. We state this formally as follows:


\begin{corollary}\label{cor:DF}
The BIC property holds even if agents have different priors, and incomplete knowledge of the said priors, as long as (i) all agents agree on the ordering of prior mean rewards (which is known to the algorithm), (ii) each agent's prior satisfies the conditions in Assumption~\ref{ass:detail-free} for the chosen parameters $(\hat{\mu}, N)$, and (iii) each agent knows that her prior satisfies (i) and (ii).
\end{corollary}

Our algorithm consists of two stages. The \emph{sampling stage} samples each arm a given number of times; it is a detail-free version of Algorithm~\ref{alg:block-m}. The racing stage is a BIC version of \emph{Active Arms Elimination}, a bandit algorithm from \citet{EvenDar-icml06}. Essentially, arms ``race" against one another and drop out over time, until only one arm remains.%
\footnote{This approach is known as \emph{Hoeffding race} in the literature, as Azuma-Hoeffding inequality is used to guarantee the high-probability selection of the best arm; see Related Work for relevant citations.}

We present a somewhat more flexible parametrization of the algorithm than in Theorem~\ref{thm:DF}, which allows for better additive constants. In particular, instead of one parameter $N$ we use several parameters, so that each parameter needs to be larger than a separate prior-dependent threshold in order to guarantee the BIC property. To obtain Theorem~\ref{thm:DF}, we simply make all these parameters equal to $N$, and let $N_\mP$ be the largest threshold.

The rest of this section is organized as follows: we present the sampling stage in Section~\ref{sec:DF-sampling}, then the racing stage in Section~\ref{sec:DF-racing}, and wrap up the proof of the main theorem in Section~\ref{sec:DF-wrapup}. The special case of two arms is substantially simpler, we provide a standalone exposition in Appendix~\ref{sec:DF-two-arms}.

\subsection{The detail-free sampling stage}
\label{sec:DF-sampling}


To build up intuition, consider the special case of two arms. Our algorithm is similar to Algorithm~\ref{alg:block}, but chooses the ``exploit" arm $a^*$ in a different way. Essentially, we use the sample average reward $\sa{\mu}_1$ of arm $1$ instead of its posterior mean reward, and compare it against the prior mean reward $\mu_2^0$ of arm $2$ (which is the same as the posterior mean of arm $2$, because the common prior is independent across arms). To account for the fact that $\sa{\mu}_1$ is only an imprecise estimate of the posterior mean reward, we pick arm $2$ as the exploit arm only if $\mu_2^0$ exceeds $\sa{\mu}_1$ by a sufficient margin. The detailed pseudocode and a standalone analysis of this algorithm can be found in Appendix~\ref{sec:DF-two-arms}.

Now we turn to the general case of $m>2$ arms. On a very high level, the sampling stage proceeds as in Algorithm~\ref{alg:block-m}: after collecting some samples from arm $1$, the process again is split into $m-1$ phases such that arm $i$ is explored in phase $i$. More precisely, rounds with arm $i$ are inserted uniformly at random among many rounds in which the ``exploit arm" $a^*_i$ is chosen.

However, the definition of the exploit arm is rather tricky, essentially because it must be made ``with high confidence" and cannot based on brittle comparisons. Consider a natural first attempt: define the ``winner" $w_i$ of the previous phases as an arm $j<i$ with the highest average reward, and pick the exploit arm to be either $w_i$ or $i$, depending on whether the sample average reward of arm $w_i$ is greater than the prior mean reward of arm $i$. Then it is not clear that a given arm, conditional on being chosen as the exploit arm, will have a higher expected reward than any other arm: it appears that no Chernoff-Hoeffding Bound argument can lead to such a conclusion. We resolve this issue by always using arm $1$ as the winner of previous phases ($w_i=1$), so that the exploit arm at phase $i$ is either arm $1$ or arm $i$. Further, we alter the condition to determine which of the two arms is the exploit arm: for each phase $i>2$, this condition now involves comparisons with the previously considered arms. Namely, arm $i$ is the exploit arm if only if the sample average reward of each arm $j\in(1,i)$ is larger than the sample average reward of arm $1$ and smaller than the prior mean reward of arm $i$, and both inequalities hold by some margin.

The pseudocode is in Algorithm~\ref{alg:block-m-df}. The algorithm has three parameters $k,L,C$, whose meaning is as follows: the algorithm collects $k$ samples from each arm, each phase lasts for $kL$ rounds, and $C$ is the safety margin in a rule that defines the exploit arm. Note that the algorithm completes in $k+Lk(m-1)$ rounds.

\begin{algorithm}[ht]
\SetKwInOut{Input}{Parameters}\SetKwInOut{Output}{Output}
\Input{$k,L\in\N$ and $C\in(0,1)$}
\BlankLine
\nl For the first $k$ agents recommend arm $1$, and let $\sa{\mu}_1^k$ be their average reward\;
\For{each $\text{arm}$ $i>1$ in increasing lexicographic order}
{
    \eIf{$\sa{\mu}_1^k < \mu_i^0-C$ and $\sa{\mu}_1^k +C < \sa{\mu}_j^k < \mu_i^0-C$ for all arms $j\in(1,i)$}
    {$a_i^*=i$}{ $a_i^*=1$}
	\nl From the set $P$ of the next $L\cdot k$ agents, pick a set $Q$ of $k$ agents uniformly at random\;
	\nl Every agent $p\in P-Q$ is recommended arm $a_i^*$\;
	\nl Every agent $p\in Q$ is recommended arm $i$\;
	\nl Let $\sa{\mu}_i^k$ be the average reward among the agents in $Q$\;
}
\caption{The detail-free sampling stage: samples each arm $k$ times.}
\label{alg:block-m-df}
\end{algorithm}

The provable guarantee for Algorithm~\ref{alg:block-m-df} is stated as follows.

\begin{lemma}\label{lm:DF-sampling}
Assume \propref{prop:DF}. Fix $C\in(0,\mu_m^0/3)$ and consider event
\begin{align*}
\Ev = \left\{ \forall j \in A-\{1,m\}: \mu_1+\tfrac{3C}{2}\leq \mu_j \leq \mu_m^0-\tfrac{3C}{2}\right\}.
\end{align*}
Algorithm~\ref{alg:block-m-df} is BIC if parameters $k,L$ are large enough:
\begin{align}\label{eq:lem:dfblockm-ic:parameters}
k \geq 8 C^{-2}\; \log\left(\frac{16m}{C\cdot \Pr{\Ev}}\right)
\qquad\text{and}\qquad
L   \geq 1+ 2\frac{\mu_1^0-\mu_m^0}{C\cdot\Pr{\Ev}}.
\end{align}
\end{lemma}

\noindent Observe that $\Pr{\Ev}>0$ by \propref{prop:DF}, using the fact that $C<\mu_m^0/3$.

\begin{proof}
Consider agent $p$ in phase $i$ of the algorithm. Recall that she is recommended either arm $1$ or arm $i$. We will prove incentive-compatibility for the two cases separately.

Let $\hat{\Ev}_i$ be the event that the exploit arm $a_i^*$ in phase $i$ is equal to $i$. By algorithm's specification,
\begin{align*}
\hat{\Ev}_i
    = \{\hat{\mu}_1^k \leq \min_{1<j<i} \hat{\mu}_j - C \text{ and } \max_{1\leq j<i}\hat{\mu}_j \leq \mu_i^0-C\}.
\end{align*}

\xhdr{Part I: Recommendation for arm $i$.} We will first argue that an agent $p$ who is recommended arm $i$ does not want to switch to any other arm $j$. The case $j>i$ is trivial because no information about arms $j$ or $i$ has been collected by the algorithm (because the prior is independent across arms), and $\mu_i^0\geq \mu_j^0$. Thus it suffices to consider $j<i$. We need to show that
\begin{align*}
\E[\mu_i^0-\mu_j | \rec{i}{p}] \Pr{\rec{i}{p}}\geq 0.
\end{align*}

Agent $p$ is in the ``explore group" $Q$ with probability $\tfrac{1}{L}$ and in the ``exploit group" $P-Q$ with the remaining  probability $1-\tfrac{1}{L}$. Conditional on being in the explore group, the expected gain from switching to arm $j$ is $\mu_i^0-\mu_j^0$, since that event does not imply any information about the rewards. If she is in the exploit group, then all that recommendation for arm $i$ implies is that event $\hat{\Ev}_i$ has happened. Thus:
\begin{align*}
\E[\mu_i^0-\mu_j|\rec{i}{p}]\Pr{\rec{i}{p}} = \E[\mu_i^0-\mu_j | \hat{\Ev}_i] \Pr{\hat{\Ev}_i}\cdot \left(1-\tfrac{1}{L}\right) +  \frac{\mu_i^0-\mu_j^0}{L}.
\end{align*}

Thus for the algorithm to be BIC we need to pick a sufficiently large parameter $L$:
\begin{align*}
L \geq 1+\frac{\mu_j^0-\mu_i^0}{\E[\mu_i^0-\mu_j|\hat{\Ev}_i]\Pr{\hat{\Ev}_i}}.
\end{align*}

It remains to lower-bound the denominator in the above equation. Consider the Chernoff event:
\begin{align*}
\mC = \{ \forall i\in A: |\hat{\mu}_i^k-\mu_i | \leq \eps \}.
\end{align*}
Observe that by Chernoff-Hoeffding Bound and the union bound:
\begin{align*}
\forall \mu_1,\ldots, \mu_m: \PrC{\mC}{\mu_1,\ldots,\mu_m} \geq 1-m\cdot 2\,e^{-2\eps^2 k}.
\end{align*}
Now observe that since $\mu_i\in [0,1]$ and by the Chernoff-Hoeffding Bound:
\begin{align*}
\E[\mu_i^0-\mu_j|\hat{\Ev}_i]\Pr{\hat{\Ev}_i} =~& \E[\mu_i^0-\mu_j|\hat{\Ev}_i,\mC]\Pr{\hat{\Ev}_i,\mC}+\E[\mu_i^0-\mu_j|\hat{\Ev}_i,\neg \mC]\Pr{\hat{\Ev}_i,\neg \mC}\\
\geq~& \E[\mu_i^0-\mu_j|\hat{\Ev}_i,\mC]\Pr{\hat{\Ev}_i,\mC}-2\,e^{-2\eps^2 k}.
\end{align*}
Observe that conditional on the events $\hat{\Ev}_i$ and $\mC$, we know that $\mu_i^0-\mu_j\geq  C-\eps$, since event $\hat{\Ev}_i$ implies that $\mu_i^0\geq \hat{\mu}_j^k+C$ and event $\mC$ implies that $\hat{\mu}_j^k\geq \mu_j-\eps$. Thus we get:
\begin{align*}
\E[\mu_i^0-\mu_j|\hat{\Ev}_i]\Pr{\hat{\Ev}_i}\geq~& (C-\eps)\Pr{\hat{\Ev}_i,\mC}-m\cdot 2\,e^{-2\eps^2 k}.
\end{align*}
Now let:
\begin{align*}
\Ev_i = \{\mu_1 \leq \min_{1<j<i} \mu_j - C -2\eps \text{ and } \max_{1\leq j<i}\mu_j \leq \mu_i^0-C-2\eps\}\;,
\end{align*}
i.e., $\Ev_i$ is a version of $\hat{\Ev}_i$ but on the true mean rewards, rather than the sample averages. Observe that under event $\mC$, event $\Ev_i$ implies event $\hat{\Ev}_i$. Thus:
\begin{equation}\label{eqn:hat-vs-true}
\Pr{\Ev_i,\mC} \leq \Pr{\hat{\Ev}_i,\mC}\;.
\end{equation}
Moreover, by Chernoff-Hoeffding Bound:
\begin{equation}\label{eqn:chernoff-and-ev_i}
\Pr{\Ev_i,\mC} = \PrC{\mC}{\Ev_i}\Pr{\Ev_i}\geq \left(1-m\cdot e^{-2\eps^2k}\right) \Pr{\Ev_i}.
\end{equation}
Hence, we conclude that:
\begin{align*}
\E[\mu_i^0-\mu_j|\hat{\Ev}_i]\Pr{\hat{\Ev}_i}\geq~& (C-\eps)\left(1-m\cdot e^{-2\eps^2k}\right) \Pr{\Ev_i}-m\cdot 2\,e^{-2\eps^2 k}.
\end{align*}
By taking $\eps=\frac{C}{4}$ and
    $k\geq 2 \eps^{-2}\, \log\left(\frac{16m}{C\cdot \Pr{\Ev_i}}\right)$
we get
that:
\begin{align*}
\E[\mu_i^0-\mu_j|\hat{\Ev}_i]\Pr{\hat{\Ev}_i}\geq \tfrac{3}{4} C \left(1- \tfrac{1}{8}\right) \Pr{\Ev_i} - \tfrac{1}{8} C \Pr{\Ev_i} \geq \tfrac{1}{2} C \Pr{\Ev_i}.
\end{align*}
Thus it suffices to pick:
\begin{align*}
L \geq 1+
2
\frac{\mu_j^0-\mu_i^0}{C \Pr{\Ev_i}}.
\end{align*}
In turn, since $\mu_j^0\leq \mu_1^0$ it suffices to pick the above for $j=1$.

\xhdr{Part II: Recommendation for arm $1$.} When agent $p$ is recommended arm $1$ she knows that she is not in the explore group, and therefore all that is implied by the recommendation is the event $\neg \Ev_i$. Thus we need to show that for any arm $j\geq 1$:
\begin{align*}
\E[\mu_1-\mu_j|\neg \Ev_i]\Pr{\neg \Ev_i}\geq 0.
\end{align*}
It suffices to consider arms $j\leq i$. Indeed, if an agent does not want to deviate to arm $j=i$ then she also does not want to deviate to any arm $j>i$, because at phase $i$ the algorithm has not collected any information about arms $j\geq i$.

Now observe the following:
\begin{align*}
\E[\mu_1-\mu_j| \neq \Ev_i]\Pr{\neg \Ev_i} =~& \E[\mu_1-\mu_j] - \E[\mu_1-\mu_j| \Ev_i]\Pr{\Ev_i}\\
=~& \mu_1^0-\mu_j^0 + \E[\mu_j-\mu_1 | \Ev_i]\Pr{\Ev_i}
\end{align*}
Since, by definition $\mu_1^0\geq \mu_j^0$, it suffices to show that for any $j\in [2,i]$:
\begin{align*}
\E[\mu_j-\mu_1 | \Ev_i]\Pr{\Ev_i}\geq 0.
\end{align*}
If $j=i$, the above inequality is exactly the incentive-compatibility constraint that we showed in the first part of the proof. Thus we only need to consider arms $j\in [2,i-1]$.

Similar to the first part of the proof, we can write:
\begin{align*}
\E[\mu_j-\mu_1 | \Ev_i]\Pr{\Ev_i}\geq\E[\mu_j-\mu_1 | \Ev_i,\mC]\Pr{\Ev_i,\mC} - 2m\,e^{-2\eps^2 k}.
\end{align*}
By similar arguments, since $\Ev_i$ implies that $\hat{\mu}_j^k \geq \hat{\mu}_1^k+C$ and since $\mC$ implies that $\mu_j\geq \hat{\mu}_j^k-\eps$ and $\hat{\mu}_1^k\geq \mu_1-\eps$, we have that the intersection of the two events implies $\mu_j\geq \mu_1 + C-2\eps$. Hence:
\begin{align*}
\E[\mu_j-\mu_1 | \Ev_i]\Pr{\Ev_i}\geq (C-2\eps)\Pr{\hat{\Ev}_i,\mC} - 2m e^{-2\eps^2 k}.
\end{align*}
By Equations \eqref{eqn:hat-vs-true} and \eqref{eqn:chernoff-and-ev_i}
\begin{align*}
\E[\mu_j-\mu_1 | \Ev_i]\Pr{\Ev_i}\geq (C-2\eps)\left(1-2m e^{-2\eps^2 k}\right) \Pr{\Ev_i} - m e^{-2\eps^2 k}.
\end{align*}
Setting $\eps$ and $k$ same way as in the first part, we have
\begin{align*}
\E[\mu_j-\mu_1 | \Ev_i]\Pr{\Ev_i}\geq \tfrac{1}{2}C\left(1-\tfrac{1}{8}\right)\Pr{\Ev_i} - \tfrac{1}{8} C \Pr{\Ev_i} = \tfrac{5}{16} C \Pr{\Ev_i}\geq 0.
\end{align*}

\xhdr{Wrapping up.}
We proved that both potential recommendations in phase $i$ are BIC if we take $\eps = C/4$ and
\begin{align*}
k \geq 8 C^{-2}\; \log\left(\frac{8m}{C\cdot \Pr{\Ev_i}}\right)
\qquad\text{and}\qquad
L   \geq 1+ 2\frac{\mu_1^0-\mu_m^0}{C\cdot\Pr{\Ev_i}}.
\end{align*}
\refeq{eq:lem:dfblockm-ic:parameters} suffices to ensure the above for all $i$ because
$\Ev = \Ev_m \subset \Ev_i$ for all phases $i<m$.
\end{proof}

\subsection{The detail-free racing stage}
\label{sec:DF-racing}

The detail-free racing stage is a BIC version of \emph{Active Arms Elimination}, a bandit algorithm from \citet{EvenDar-icml06}. Essentially, arms ``race" against one another and drop out over time, until only one arm remains. More formally, the algorithm proceeds as follows. Each arm is labeled ``active" initially, but may be permanently ``deactivated" once it is found to be sub-optimal with high confidence. Only active arms are chosen. The time is partitioned into phases: in each phase the set of active arms is recomputed, and then all active arms are chosen sequentially. To guarantee incentive-compatibility, the algorithm is initialized with some samples of each arm, collected in the sampling stage. The pseudocode is in Algorithm~\ref{alg:many-arms-elim}.


\begin{algorithm}[ht]
\SetKwInOut{Input}{Input}\SetKwInOut{Output}{Output}
\Input{parameters $k\in\N$ and $\theta\geq 1$; time horizon $T$.}
\Input{$k$ samples of each arm $i$, denoted $r_i^1 \LDOTS r_i^k$.}
\BlankLine
\nl Initialize the set of active arms: $B= \{\text{all arms}\}$\;
\nl Split the remainder into consecutive phases, starting from phase $n=k$\;
\Repeat{$|B|= 1$}{
    \tcp{Phase $n$}
    \nl Let $\sa{\mu}_i^n= \frac{1}{n}\sum_{t=1}^{n}r_i^t$ be the sample average for each arm $i\in B$\;
    \nl Let $\sa{\mu}_*^n=\max_{i\in B} \sa{\mu}_i^n$ be the largest average reward so far\;
    \nl Recompute $B=\left\{\text{arms $i$}: \sa{\mu}_*^n-\sa{\mu}_i^n\leq \sqrt{\frac{\log(T\theta)}{n}}\right\}$\;
	\nl To the next $|B|$ agents, recommend each arm $i\in B$ sequentially in lexicographic order\;
	\nl Let $r_i^{n+1}$ be the reward of each arm $i\in B$ in this phase\;
	\nl $n= n+1$\;
}
\nl For all remaining agents recommend the single arm $a^*$ that remains in $B$
\caption{The detail-free racing stage.}
\label{alg:many-arms-elim}
\end{algorithm}

\begin{lemma}[BIC]\label{lm:DF-racing}
Fix an absolute constant $\tau\in (0,1)$ and let
\begin{align*}
    \theta_\tau = \frac{4 m^2}{\tau}\Big/ \min_{i\in A} \Pr{\mu_i-\max_{j\neq i}\mu_j\geq \tau}.
\end{align*}
Algorithm~\ref{alg:many-arms-elim} is BIC if \propref{prop:DF} holds and the parameters satisfy
    $\theta\geq \theta_\tau$
and
    $k\geq \theta^2_\tau\, \log(T)$.
\end{lemma}

\begin{proof}
We will show that for any agent $p$ in the racing stage and for any two arms $i,j \in A$:
\begin{align*}
\E[ \mu_i - \mu_j | \rec{i}{p}]\Pr{\rec{i}{p}} \geq 0.
\end{align*}
In fact, we will prove a stronger statement:
\begin{align}\label{eq:lem:elimination-ic-app:0}
\E[ \mu_i - \max_{j\neq i} \mu_j | \rec{i}{p}] \Pr{\rec{i}{p}} \geq 0.
\end{align}
For each pair of arms $i,j\in A$, denote $Z_{ij}=\mu_i-\mu_j$ and
\[Z_i=\mu_i-\max_{j\neq i} \mu_j = \min_{j\neq i} Z_{ij}.\]
For each phase $n$, let $\sa{z}_{ij}^n = \sa{\mu}_i^n-\sa{\mu}_j^n$. For simplicity of notation, we assume that samples of eliminated arms are also drawn at each phase but simply are not revealed or used by the algorithm.

Let
    $c_n = \sqrt{\frac{\log(T\theta)}{n}}$
be the threshold in phase $n$ of the algorithm. Consider the ``Chernoff event":
\begin{align*}
\mC = \left\{ \forall n\geq k, \forall i,j\in A: |Z_{ij}-\sa{z}_{ij}^n| < c_n \right\}.
\end{align*}
By Chernoff-Hoeffding Bound and the union bound, for any $Z$, we have:
\begin{align*}
\Pr{\neg \mC | Z}
    &\leq \sum_{n\geq n^*} \sum_{i,j\in A}\Pr{|Z_{ij}-\sa{z}_{ij}^n|\geq c_n n}
    \leq \sum_{n\geq n^*} \sum_{i,j\in A} e^{-2n\cdot c_n^2}\\
    &\leq  T\cdot m^2\cdot  e^{- 2\log(T\theta)} \leq \frac{m^2}{\theta}.
\end{align*}
Given the above and since $Z_{i}\geq -1$ we can write:
\begin{align}
\E[Z_{i}|\rec{i}{p}]\Pr{\rec{i}{p}}
    =~& \E[Z_{i}|\rec{i}{p}, \mC]\Pr{\rec{i}{p}, \mC} + \E[Z_{i}|\rec{i}{p}, \neg \mC]\Pr{\rec{i}{p},\neg \mC}
        \nonumber\\
\geq~&  \E[Z_{i}|\rec{i}{p}, \mC]\Pr{\rec{i}{p}, \mC}  - \frac{m^2}{\theta}.
    \label{eq:lem:elimination-ic-app:1}
\end{align}
From here on, we will focus on the first term in \refeq{eq:lem:elimination-ic-app:1}. Essentially, we will assume that Chernoff-Hoeffding Bound hold deterministically.

Consider phase $n$. Denote
    $n^*=\theta^2_\tau\, \log(T)$,
and recall that $n\geq k\geq n^*$.

We will split the integral
    $\E[Z_{i}|\rec{i}{p}, \mC]$
into cases for the value of $Z_i$. Observe that by the definition of $c_n$ we have
    $c_n\leq c_{n^*} \leq \frac{\tau \Pr{Z_i\geq \tau}}{4}$.
Hence, conditional on $\mC$, if $Z_i\geq \tau \geq  2 c_{n^*}$, then we have definitely eliminated all other arms except arm $i$ at phase $n=k\geq n^*$, since
\[\sa{z}_{ij}^n> Z_{ij}- c_n \geq Z_i- c_n\geq c_n.\]
Moreover, if $Z_i\leq - 2c_n$, then by similar reasoning, we must have already eliminated arm $i$, since
\[\sa{\mu}_i^n-\sa{\mu}_*^n = \min_{j\neq i} \sa{z}_{ij}^n < \min_{j\neq i} Z_{ij} +c_n =Z_i+c_n \leq -c_n.\]
Thus in that case $\rec{i}{p}$, cannot occur. Moreover, if we ignore the case when $Z_i\in [0,\tau)$, then the integral can only decrease, since the integrand is non-negative. Since we want to lower bound it, we will ignore this region. Hence:
\begin{align}
\E[Z_i|\rec{i}{p}, \mC]\Pr{\rec{i}{p}, \mC}
\geq~& \tau \Pr{\mC, Z_i\geq \tau}-2\cdot c_n \Pr{\mC, -2c_n\leq Z_i\leq 0}
    \nonumber \\
\geq~& \tau \Pr{\mC, Z_i\geq \tau}-2\cdot c_n\\
\geq~& \tau \Pr{\mC, Z_i\geq \tau}-\frac{\tau \cdot \Pr{Z_i\geq \tau}}{2}.
    \label{eq:lem:elimination-ic-app:2}
\end{align}
Moreover, we can lower-bound the term $\Pr{\mC, Z_i\geq \tau}$ using Chernoff-Hoeffding Bound:
\begin{align*}
\Pr{\mC, Z_i\geq \tau} =~& \Pr{\mC|Z_i\geq \tau} \Pr{Z_i\geq \tau}
\geq \left(1-\sum_{t\geq n^*}\sum_{i,j\in A} e^{-2 c_n^2 t}\right)\cdot  \Pr{Z_i\geq \tau}\\
\geq~& \left(1-m^2/\theta\right)\cdot  \Pr{Z_i\geq \tau} \geq \tfrac{3}{4} \Pr{Z_i\geq \tau}.
\end{align*}
Plugging this back into \refeq{eq:lem:elimination-ic-app:2}, we obtain:
\begin{align*}
\E[Z_i|\rec{i}{p}, \mC]\Pr{\rec{i}{p}, \mC}
\geq \tfrac{\tau}{4}\cdot \Pr{\mC, Z_i\geq \tau}.
\end{align*}
Plugging this back into \refeq{eq:lem:elimination-ic-app:1} and using the fact that $\theta\geq \theta_\tau$, we obtain
    \refeq{eq:lem:elimination-ic-app:0}.
\end{proof}

\begin{lemma}[Regret]\label{lm:DF-regret-body}
Algorithm~\ref{alg:many-arms-elim} with any parameters $k\in\N$  and $\theta\geq 1$ achieves ex-post regret
\begin{align}\label{eq:lm:reg-bound-many-1}
R(T) \leq \sum_{\text{arms $i$}} \frac{18 \log(T\theta)}{\mu^*-\mu_i},
\quad \text{where $\mu^* = \max_{i\in A} \mu_i$}.
\end{align}
Denoting $\Delta = \min_{i\in A} \mu^*-\mu_i$, and letting $n^*$ be the duration of the initial sampling stage, the ex-post regret for each round $t$ of the entire algorithm (both stages) is
\begin{align}\label{eq:lm:reg-bound-many-2}
R(t)
    &\leq \min\left( n^*\Delta + \frac{18\,m\, \log(T\theta)}{\Delta},\; t\Delta \right)
    \leq n^* + O\left( \sqrt{m\,t\log (T\theta)} \right).
\end{align}
\end{lemma}

The ``logarithmic" regret bound \eqref{eq:lm:reg-bound-many-1} is proved via standard techniques, e.g., see \citet{EvenDar-icml06}. We omit the details from this version. To derive the corollary \eqref{eq:lm:reg-bound-many-2}, observe that the ex-post regret of the entire algorithm (both stages) is at most that of the second stage plus $\Delta$ per each round of the first stage; alternatively, we can also upper-bound it by $\Delta$ per round.

\OMIT{This gives the first inequality in \eqref{eq:lm:reg-bound-many-2}. To obtain the second inequality, note that it is trivial for $t\leq n^*$, and for $t>n^*$ consider consider three cases, depending on the value of $\Delta$:
\begin{align*}
\begin{cases}
R(t) \leq 2n^*                            & \text{ if }\Delta\geq \sqrt{\beta/n^*}, \\
R(t) \leq \beta/\delta\leq \sqrt{\beta t} & \text{ if }\sqrt{\beta/t} \leq \Delta \leq \sqrt{\beta/n^*},\\
R(t) \leq \Delta t \leq \sqrt{\beta t}    & \text{ if }\Delta\leq \sqrt{\beta/t}. \qquad\qquad\qquad \qedhere
\end{cases}
\end{align*}}

\subsection{Wrapping up: proof of Theorem~\ref{thm:DF}}
\label{sec:DF-wrapup}

To complete the proof of Theorem~\ref{thm:DF}, it remains to transition to the simpler parametrization, using only two parameters
    $\hat{\mu} \in [\mu^0_m, 2\mu^0_m]$
and
    $N\geq N_\mP$,
where $N_\mP$ is a prior-dependent constant.

Let us define the parameters for both stages of the algorithm as follows. Define parameter $C$ for the sampling stage as
    $C = \hat{\mu}/6$, and set the remaining parameters $k^*,L$ for the sampling stage and
$k,\theta$ for the racing stage be equal to $N$. To define $N_\mP$, consider the thresholds for $k^*,L$ in Lemma~\ref{lm:DF-sampling} with $C = \mu^0_m/6$, and the thresholds for $k,\theta$ in Lemma~\ref{lm:DF-racing} with any fixed $\tau\in(0,1)$, and let $N_\mP$ be the largest of these four thresholds. This completes the algorithm's specification.

It is easy to see that both stages are BIC. (For the sampling stage, this is because $\Pr{\Ev}$ is monotonically non-decreasing in parameter $C$.) Further, the sampling stage lasts for $f(N) = N+N^2(m-1)$ rounds. So the regret bound in the theorem follows from that in Lemma~\ref{lm:DF-regret-body}.



\section{Extensions: contextual bandits and auxiliary feedback}
\label{sec:general}


We now extend the black-box reduction result in two directions: contexts and auxiliary feedback. Each agent comes with a signal (\emph{context}), observable by both the agent and the algorithm, which impacts the rewards received by this agent as well as the agent's beliefs about these rewards. Further, after the agent chooses an arm, the algorithm observes not only the agent's reward but also some auxiliary feedback. As discussed in Introduction and Related Work, each direction is well-motivated in the context of BIC exploration, and is closely related to a prominent line of work in the literature on multi-armed bandits.

\subsection{Setting: BIC contextual bandit exploration with auxiliary feedback}

\newcommand{\Dfb}{\mD_{\mathtt{fb}}}  
\newcommand{\DX}{\mD_{\mathtt{X}}}    
\newcommand{\muvec}{\vec{\mu}}       

In each round $t$, a new agent arrives, and the following happens:
\begin{OneLiners}
\item context $x_t\in \mX$ is observed by both the algorithm and the agent,
\item algorithm recommends action $a_t \in A$ and computes prediction $\phi_t$ (invisible to agents),
\item reward $r_t(a_t)\in \R$ and auxiliary feedback $f_t(a_t) \in\mF$ are observed by the algorithm.
\end{OneLiners}
The sets $\mX,A,\mF$ are fixed throughout, and called, resp., the \emph{context space}, the \emph{action space}, and the \emph{feedback space}. Here $r_t(\cdot)$ and $f_t(\cdot)$ are functions from the action space to, resp., $\R$ and $\mF$.

We assume an IID environment. Formally, the tuple
    $(x_t\in \mX;\; r_t: A\to \R;\; f_t:A\to \mF) $
is chosen independently from some fixed distribution $\Psi$ that is not known to the agents or the algorithm.

There is a common prior on $\Psi$, denoted $\mP$, which is known to the algorithm and the agents. Specifically, we will assume the following parametric structure. The context $x_t$ is an independent draw from some fixed distribution $\DX$ over the context space. Let $\mu_{a,x}$ denote the expected reward corresponding to context $x$ and arm $a$. There is a single-parameter family $\mD(\mu)$ of reward distributions, parameterized by the mean $\mu$. The vector of expected rewards
    $\muvec = \left( \mu_{a,x}:\, a\in A, x\in \mX \right)$
is drawn according to a prior $\muPrior$. Conditional on $\muvec$ and the context $x=x_t$, the reward $r_t(a)$ of every given arm $a=a_t$ is distributed according to $\mD(\mu_{a,x})$. Similarly, there is a single-parameter family of distributions $\Dfb(\cdot)$ over the feedback space, so that conditional on $\muvec$ and the context $x=x_t$, the feedback $f_t(a)$ for every given arm $a=a_t$ is distributed according to $\Dfb(\mu_{a,x})$. In what follows, we will sometimes write
    $\mu(a,x)$
instead of $\mu_{a,x}$.

The incentive constraint is now conditioned on the context. To state it formally, the planner's recommendation in round $t$ is denoted with $I_t^x$, where $x = x_t$ is the context (recall that the recommendation can depend on the observed context). As in Definition~\ref{def:BIC}, let $\mE_{t-1}$ be the event that the agents have followed the algorithm's recommendations up to (and not including) round $t$.

\begin{definition}[Contextual BIC]
A recommendation algorithm is \emph{Bayesian incentive-compatible} (\emph{BIC}) if
\begin{align*}
\E[\mu_{a,x} \mid I_t^x=a, \mE_{t-1}] \geq \max_{a'\in A} \E[\mu_{a',x} \mid I_t^x=a, \mE_{t-1}]
    \qquad \forall t\in [T],\, \forall x\in \mX, \, \forall a\in A.
\end{align*}
\end{definition}

\xhdr{Performance measures.}
Algorithm's performance is primarily measured in terms of \emph{contextual regret}, defined as follows. A \emph{policy} $\pi:\mX \to A$ is a mapping from contexts to actions. The number of possible policies increases exponentially in the cardinality of the context space. Since in realistic applications the context space tends to be huge, learning over the space of all possible policies is usually intractable. A standard way to resolve this (following \citet{Langford-nips07}) is to explicitly restrict the class of policies. Then the algorithm is given a set $\Pi$ of policies, henceforth, the \emph{policy set}, so that the algorithm's performance is compared to the best policy in $\Pi$. More specifically, Bayesian contextual regret of algorithm $\ALG$ relative to the policy class $\Pi$ in the first $t$ rounds is defined as
\begin{align}\label{eq:context-Bayesian-regret-defn}
R_{\ALG,\Pi}(t)
    = \E_{\muvec\sim\muPrior}\left[
            t\cdot \sup_{\pi\in \Pi}\; \E_{x\in \DX} \left[ \mu(\pi(x),\,x) \right]
            - \E_{\ALG}\left[ \sum_{s=1}^{t} \mu(I_s^x,\,x)\right]
    \right].
\end{align}

As in Section~\ref{sec:black-box}, we will also provide performance guarantees in terms of average rewards and in terms of the quality of predictions.

\xhdr{Discussion.}
In the line of work on contextual bandits with policy sets, a typical algorithm is in fact a family of algorithms, parameterized by an oracle that solves a particular optimization problem for a given policy class $\Pi$. The algorithm is then oblivious to how this oracle is implemented. This is a powerful approach, essentially reducing contextual bandits to machine learning on a given data set, so as to take advantage of the rich body of work on the latter. In particular, while the relevant optimization problem is NP-hard for most policy classes studied in the literature, it is often solved very efficiently in practice.

Thus, it is highly desirable to have a BIC contextual bandit algorithm which can work with different oracles for different policy classes. However, an explicit design of such algorithm would necessitate BIC implementations for particular oracles, which appears very tedious. A black-box reduction such as ours circumvents this difficulty.

\xhdr{A note on notation.}
We would like to re-order the actions according to their prior mean rewards, as we did before. However, we need to do it carefully, as this ordering is now context-specific. Let
    $\mu_{a,x}^0=\E[\mu_{a,x}]$
be the prior mean reward of arm $a$ given context $x$. Denote with $\sigma(x,i)\in A$ the $i$-th ranked arm for $x$ according to the prior mean rewards:
\begin{align*}
    \mu_{\sigma(x,1),\,x}^0 \geq \mu_{\sigma(x,2),\,x}^0 \geq \ldots \geq \mu_{\sigma(x,m),\,x}^0,
\end{align*}
where $m$ is the number of arms and ties are broken arbitrarily. The \emph{arm-rank} of arm $a$ given context $x$ is the $i$ such that $a=\sigma(x,i)$. We will use arm-ranks as an equivalent representation of the action space: choosing arm-rank $i$ corresponds to choosing action $\sigma(x,i)$ for a given context $x$.

A \emph{rank-sample} of arm-rank $i$ is a tuple $(x,a,r,f)$, where $x$ is a context, $a = \sigma(x,i)$ is an arm, and $r,f$ are, resp., the reward and auxiliary feedback received for choosing this arm for this context. Unless the $x$ is specified explicitly, it is drawn from distribution $\DX$.



\subsection{The black-box reduction and provable guarantees}

Given an arbitrary algorithm $\ALG$, the black-box reduction produces an algorithm $\ALGIC$. It proceeds in two stages: the sampling stage and the simulation stage, which extend the respective stages for BIC bandit exploration. In particular, the sampling stage does not depend on the original algorithm $\ALG$.

The main new idea is that the sampling stage considers arm-ranks instead of arms: it proceeds in phases $i=1,2,\,\ldots\;$, so that in each phase $i$ it collects samples of arm-rank $i$, and the exploit arm is (implicitly) the result of a contest between arm-ranks $j<i$. Essentially, this design side-steps the dependence on contexts.

Each stage of the black-box reduction has two parameters, $k,L\in\N$. While one can set these parameters separately for each stage, for the sake of clarity we consider a slightly suboptimal version of the reduction in which the parameters are the same for both stages. The pseudo-code is given in Algorithm~\ref{alg:BB-contextual}. The algorithm maintains the ``current" dataset $\mS$ and keeps adding various rank-samples to this dataset. The ``exploit arm" now depends on the current context, maximizing the posterior mean reward given the current dataset $\mS$.

\begin{algorithm}[ht]
\SetKwInOut{Input}{Parameters}\SetKwInOut{Output}{Output}
\Input{$k,L\in\N$; contextual bandit algorithm $\mA$}
\BlankLine
 \tcp{Sampling stage: obtains $k$ samples from each arm-rank}
\nl Initialize the dataset: $\mS=\emptyset$\;
\nl For each context $x\in \mX$: denote
    $a_x^*(\mS) = \arg\max_{a \in A}  \E[\mu_{a,x}| \mS]$\;	
\nl Recommend arm-rank $1$ to the first $k$ agents,
     add the rank-samples to $\mS$ \;
\For{each phase $i=2$ to $m$}
{
     \tcp{each phase lasts $KL$ rounds}
	\nl From the set $P$ of the next $kL$ agents, pick a set $Q$ of $k$ agents uniformly at random\;
	\nl For each agent $p\in P-Q$, recommend the ``exploit arm" $a_{x_p}^*(\mS)$\;
	\nl For each agent in $Q$, recommend arm-rank $i$\;
	\nl Add to $\mS$ the rank-samples from all agents in $Q$
}
\tcp{Simulation stage (each phase lasts $L$ rounds)}
\For{each phase $n = 1,2,\,\ldots\;$}{
\nl Pick one agent $p$ from the set $P$ of the next $L$ agents uniformly at random\;
\nl For agent $p$: send context $x_p$ to algorithm $\mA$, get back arm $a_n$ and prediction $\phi_n$\;
\nl \algTab recommend arm $a_n$ to agent $p$\;
\nl For every agent $t\in P\setminus \{p\}$: recommend the ``exploit arm" $a_{x_t}^*(\mS)$\;
\nl For every agent $t\in P$ in phase $n\geq 2$: record prediction $\phi_{n-1}$\;
\nl Return the rank sample from agent $p$ to algorithm $\mA$\;
\nl Add to $\mS$ the rank-samples from all agents in this phase\;
}
\caption{Black-box reduction with contexts, auxiliary feedback and predictions.}
\label{alg:BB-contextual}
\end{algorithm}

The black-box accesses the common prior via the following two oracles. First, given a context $x$, rank all actions according to their prior mean rewards $\mu^0_{a,x}$. Second, given a context $x$ and a set $S$ of rank-samples, compute the action that maximizes the posterior mean, i.e.,
    $\arg \max_{a\in A} \E [\mu_{a,x}|S]$.

As before, we need to restrict the common prior $\mP$ to guarantee incentive-compatibility. Our assumption talks about posterior mean rewards conditional on several rank-samples. For a particular context $x$ and arm-rank $i$, we are interested in the smallest difference, in terms of posterior mean rewards, between the arm-rank $i$ and any other arm. We state our assumption as follows:

\begin{property}
\item\label{prop:contextual}
For a given arm-rank $i\in[m]$ and parameter $k\in\N$, let $\Lambda_i^k$ be a random variable representing $k_i\geq k$ rank-samples of arm-rank $i$. Given context $x$, denote
\begin{align*}
X_{(i,j,x)}^k =
    \min_{\text{arms}\;a\neq \sigma(x,i)}\;
        \E\left[ \mu_{\sigma(x,i),x}-\mu_{a,x}  | \Lambda_1^k,\ldots, \Lambda_j^k \right].
\end{align*}
There exist prior-dependent constants $k_\mP<\infty$ and $\tau_\mP, \rho_\mP>0$ such that
\begin{align*}
\Pr[x\sim \DX]{X_{(i,j,x)}^k> \tau_\mP} \geq \rho_{\mP}
\end{align*}
for any $k\geq k_\mP$, any arm-rank $i\in [m]$, and $j\in\{i-1,m\}$.
\end{property}

The analysis is essentially the same as before. (For the sampling stage, the only difference is that the whole analysis for any single agent is done conditional on her context $x$ and uses arm-ranks instead of arms.)

\propref{prop:contextual} is used with $j=i-1$ for the sampling stage, and with $j=m$ for the simulation stage. For each stage, the respective property is discovered naturally as the condition needed to complete the proof.

\begin{theorem}[BIC]\label{thm:contextual-reduction-BIC}
Assume the common prior $\mP$ satisfies \propref{prop:contextual} with constants $(k_\mP, \tau_\mP,\rho_\mP)$. The black-box reduction is BIC (applied to any algorithm) if its parameters satisfy $k\geq k_\mP$ and $L\geq L_\mP$, where $L_\mP$ is a finite prior-dependent constant. In particular, one can take
\begin{align*}
L_\mP = 1+\sup_{x\in \mX}\max_{a,a'\in A}\frac{\mu_{a,x}^0-\mu_{a',x}^0}{\tau_\mP \cdot \rho_\mP}.
\end{align*}
\end{theorem}

The performance guarantees are similar to those in Theorem~\ref{thm:reduction-performance}.

\begin{theorem}[Performance]\label{thm:contextual-reduction-performance}
Consider the black-box reduction with parameters $k,L$, applied to some algorithm $\ALG$. The sampling stage then completes in
    $c = mLk+k$
rounds, where $m$ is the number of arms.
\begin{itemize}
\item[(a)]
Let $\widehat{U}_{\ALG}(\tau_1,\tau_2)$ be the average Bayesian-expected reward of algorithm $\ALG$ in rounds
$[\tau_1,\tau_2]$. Then
\begin{align*}
    \widehat{U}_{\ALGIC}(c+1,c+\tau) \geq  \widehat{U}_{\ALG}(1,\flr{\tau/L})
    \qquad\text{for any duration $\tau$}.
\end{align*}


\item[(c)] Let $R_{\ALG,\Pi}(t)$ be the Bayesian contextual regret of algorithm $\ALG$ in the first $t$ rounds, relative to some policy class $\Pi$. Then:
\begin{align*}
R_{\ALGIC,\,\Pi}(t) \leq L\cdot R_{\ALG,\,\Pi}\left( \flr{t/L} \right)
        + c\cdot \E_{x\sim\DX}\left[ \max_{a\in A} \mu_{a,x} - \min_{a\in A} \mu_{a,x} \right].
\end{align*}
In particular, if $\mu_{a,x} \in [0,1]$ and the original algorithm $\ALG$ achieves asymptotically optimal Bayesian contextual regret $O(\sqrt{t \log|\Pi|})$, then so does $\ALGIC$ (assuming $L$ is a constant).

\item[(c)] The prediction returned by algorithm $\ALGIC$ in each round $t>c+L$ has the same distribution as that returned by algorithm $\ALG$ in round $\flr{(t-c)/L}$.

\end{itemize}
\end{theorem}

\OMIT{We make technical assumptions about the prior distributions. Intuitively, we consider the case that context is fairly likely to repeat often, and in such cases our assumptions would give interesting performance guarantees.
The incentive compatibility would always be guaranteed.}

\section{Properties of the common prior}
\label{sec:properties}

\propref{prop:general-persuasion}, our main assumption on the common prior for the black-box reduction, is stated in a rather abstract manner for the sake of generality (e.g., so as to allow correlated priors), and in order to avoid some excessive technicalities. In this section we discuss some necessary or sufficient conditions. First, we argue that \emph{conditioning on more samples can only help}, i.e. that if a similar property holds conditioning on less samples, it will still hold conditioning on more samples. This is needed in the incentives analysis of the simulation stage, and to show the ``almost necessity" for two arms. Second, we exhibit some natural sufficient conditions if the prior is independent across arms, and specifically prove Lemma~\ref{lm:sufficient-conditions-on-P}.

\subsection{Conditioning on more samples can only help}
\label{sec:properties-more}

Fix arm $i$ and vector $\vec{a} = (a_1 \LDOTS a_m)\in \N^m$, where $m$ is the number of arms. Let $\Lambda_{\vec{a}}$ be a random variable representing the first $a_j$ samples from each arm $j$. Denote
\begin{align*}
X_{i,\vec{a}} = \min_{\text{arms $j\neq i$}}\; \E\left[\mu_i-\mu_{j}|\Lambda_{\vec{a}} \right].
\end{align*}
We are interested in the following property:
\begin{property}
\item \label{prop:switch-generic}
There exist prior-dependent constants $\tau_\mP, \rho_\mP>0$ such that
    $\Pr{X_{i,\vec{a}} \geq \tau_\mP} \geq \rho_\mP$.
\end{property}

\noindent We argue that \propref{prop:switch-generic} stays true if we condition on more samples.

\begin{lemma}
Fix arm $i$ and consider vectors $\vec{a},\vec{b}\in \N^m$ such that $a_j\leq b_j$ for all arms $j$. If \propref{prop:switch-generic} holds for arm $i$ and vector $\vec{a}$, then it also holds for arm $i$ and vector $\vec{b}$.
\end{lemma}

\begin{proof}
For the sake of contradiction, assume that the property for this arm holds for $\vec{a}$, but not for $\vec{b}$. Denote
    $X = X_{i,\vec{a}}$ and $Y = X_{i,\vec{b}}$.

Then, in particular, we have
    $\Pr{Y>1/\ell^2}< 1/\ell^2$
for each $\ell\in\N$. It follows that $\Pr{Y\leq 0}=1$. (To prove this, apply Borel-Cantelli Lemma with sets $S_\ell=\{Y>1/\ell^2\}$, noting that $\sum_\ell \Pr{S_\ell}$ is finite and $\limsup_\ell S_\ell= \{Y>0\}$.)

The samples in $\Lambda_{\vec{a}}$ are a subset of those for $\Lambda_{\vec{b}}$, so
    $\sigma(\Lambda_{\vec{a}}) \subset \sigma(\Lambda_{\vec{b}})$,
and consequently
    $X = \E[Y | \Lambda_{\vec{a}}]$.
Since $Y\leq 0$ a.s., it follows that $X\leq 0$ a.s. as well, contradicting our assumption.
\end{proof}

Therefore, \propref{prop:general-persuasion} implies \propref{prop:switch} with the same constant $k_\mP$; this is needed in the incentives analysis of the simulation stage. Also, \propref{prop:persuasion} implies \propref{prop:general-persuasion} with the same $k_\mP$, so the latter property is (also) ``almost necessary" for two arms, in the sense that it is necessary for a strongly BIC algorithm (see Lemma~\ref{lm:propA-app}).

\subsection{Sufficient conditions: Proof of Lemma~\ref{lm:sufficient-conditions-on-P}}

For each arm $i$, let $S_i^k$ be a random variable that represents $n$ samples of arm $i$, and let
    $Y_i^k= \E[\mu_i | S_i^k]$
be the posterior mean reward of this arm conditional on this random variable.

Since the prior is independent across arms, \refeq{eq:prop:general-persuasion} in \propref{prop:general-persuasion} can be simplified as follows:
\begin{align}\label{eq:lm:sufficient-conditions-on-P-1}
\Pr{\mu_i^0 - \max \left\{ Y_1^k \LDOTS Y_{i-1}^k;\, \mu_{i+1}^0 \right\}  \geq \tau_\mP}\geq \rho_{\mP}
    \qquad\forall k>k_\mP, i\in A.
\end{align}

Fix some $\tau_\mP$ such that $0< \tau_\mP < \min_{j\in A} (\mu_j^0 - \mu_{j+1}^0)$. Note that the right-hand side is strictly positive because the prior mean rewards are all distinct (condition (iii) in the lemma).

Since
    $(\mu_i:\, i\in A)$
are mutually independent, and for fixed $k$ each $S_i^k$, $i\in A$ is independent of everything else conditional on $\mu_i$, we have that $(Y_i^k:\, i\in A)$ are mutually independent. Thus, one can rewrite \refeq{eq:lm:sufficient-conditions-on-P-1} as follows:
\begin{align*}
\prod_{j< i} \Pr{Y_j^k\leq \mu_i^0 -\tau_\mP} \geq \rho_{\mP}
    \qquad\forall k>k_\mP, i\in A.
\end{align*}
Therefore it suffices to prove that $\tau_\mP$ as fixed above, and some prior-dependent constant $k_\mP$, for any two arms $j< i$ we have:
\begin{align}\label{eq:lm:sufficient-conditions-on-P-2}
\Pr{Y_j^k\leq \mu_i^0 -\tau_\mP}\geq q>0
\qquad \forall k>k_\mP.
\end{align}
where $q$ is some prior-dependent constant.
Observe that event $\{Y_j^k\leq \mu_i^0 -\tau_\mP\}$ is implied by the event
    $\{Y_j^k-\mu_j\leq \eps \} \cap \{ \mu_j\leq \mu_i^0- \tau_{\mP}-\eps\}$,
where $\eps>0$ is a prior-dependent constant to be fixed later.
Invoking also the union bound we get:
\begin{align}
\Pr{ Y_j^k\leq \mu_i^0 -\tau_\mP}
    &\geq \Pr{ (Y_j^k-\mu_j\leq \eps) \cap (\mu_j\leq \mu_i^0- \tau_{\mP}-\eps)}
        \nonumber\\
    &\geq \Pr{ \mu_j\leq \mu_i^0- \tau_{\mP}-\eps} -  \Pr{Y_j^k-\mu_j>\eps}.
\label{eq:lm:sufficient-conditions-on-P-3}
\end{align}
By the full support assumption (condition (ii) in the lemma), the first summand in \eqref{eq:lm:sufficient-conditions-on-P-3} is bounded from below by some constant $\rho>0$  when
    $\mu_i^0 - \tau_{\mP}-\eps>b$.
The latter is achieved when $\tau_{\mP}, \eps$ are sufficiently small
(because $\mu_i^0>a$, also by the full support assumption). Fix such $\tau_{\mP}, \eps$ from here on.

It suffices to bound the second summand in \eqref{eq:lm:sufficient-conditions-on-P-3} from above by $\rho/2$. This follows by a convergence property \refprop{prop:bayes-convergence} which we articulate below. By this property, there exists a prior-dependent constant $k^*$ such that $\Pr{Y_j^k-\mu_j >\eps} \leq \rho/2$ for any $k>k^*$. Thus, we have proved
    \refeq{eq:lm:sufficient-conditions-on-P-2}
for $k_\mP = k^*$ and $q = \rho/2$. This completes the proof of the lemma.


It remains to state and prove the convergence property used above.

\begin{property}
\item \label{prop:bayes-convergence}
For each arm $i$, $Y_i^k$ converges to $\mu_i$ in probability as $k\to\infty$. That is,
\begin{align*}
\forall i\in A\quad
\forall\eps,\delta>0 \quad
\exists k^*(\eps,\delta)<\infty \quad
\forall k\geq k^*(\eps,\delta) \quad
        \Pr{ \left|Y_i^k - \mu_i \right| > \eps}\leq \delta.
\end{align*}
\end{property}

This property follows from the convergence theorem in \citet{Doob49}, because the realized rewards are bounded (condition (iv) in the lemma). For a more clear exposition of the aforementioned theorem, see e.g., Theorem 1 in \citet{Ghosal-review96}. There the sufficient condition is that the parameterized reward distribution $\mD_i(\mu_i)$ of each arm $i$ is pointwise dominated by an integrable random variable; the latter is satisfied if the realized rewards are bounded.

\section{Conclusions and open questions}
\label{sec:conclusions}

We have resolved the asymptotic regret rates achievable with incentive compatible exploration for a constant number of actions, as outlined in the Introduction. We provided an algorithm for BIC bandit exploration with asymptotically optimal regret (for a constant number of actions), and a general ``black-box" reduction from arbitrary bandit algorithms to BIC ones that works in a very general explore-exploit setting and increases the regret by at most a constant multiplicative factor.



This paper sets the stage for future work in several directions. First, the most immediate technical questions left open by our work is whether one can achieve Bayesian regret with constants that do not depend on the prior (for a fixed $m$, the number of actions), and with only a polynomial dependence on $m$. Second, one would like to handle bandit problems with a large number of actions and a known structure on the action set.%
\footnote{Structured action sets have been one of the major directions in the literature on MAB in the past decade. Several structures have been studied, e.g., linearity, convexity, Lipschitzness, and combinatorial structure; see \citet{Bubeck-survey12} for background and references.}
This would require handling agents' priors with complex correlations across actions, which do not necessarily satisfy the properties assumed in this paper. On the other hand, algorithms for such problems typically explore only a sparse subset of the arms, so perhaps the assumptions on the prior could be relaxed. 
Third, the mechanism design setup could be generalized to incorporate constraints on how much information about the previous rounds must be revealed to the agents. Such constraints, arising for practical reasons in the Internet economy and for legal reasons in medical decisions, typically work against the information asymmetry, and hence make the planner's problem more difficult. Fourth, while our detail-free result does not require the algorithm to have full knowledge of the priors, ideally the planner would like to start with little or no knowledge of the priors, and elicit the necessary knowledge directly from the agents. Then the submitted information becomes subject to agent's incentives, along with the agent's choice of action. Finally, it would be tantalizing to bring the theory of BIC exploration closer to the design of large-scale medical trials (see Section~\ref{sec:medical} for discussion). In particular, one needs to account for heterogenous, risk-averse patients whose beliefs and risk preferences are largely unknown (and probably need to be elicited), and whose decisions may deviate from the standard economic assumptions.


\ACKNOWLEDGMENT{
The authors are grateful to Johannes H\"{o}rner and Yeon-Koo Che for insightful discussions, and to Sebastien Bubeck for a brief collaboration during the initial stage of this project.
}

\bibliographystyle{plainnat}
\bibliography{bib-abbrv,bib-bandits,bib-AGT,bib-slivkins,bib-crowdsourcing,bib-medical,bib-ML,bib-random,bib-fairness}

\newpage
\begin{APPENDICES}
\SingleSpacedXI

\section{A simpler detail-free algorithm for two arms}
\label{sec:DF-two-arms}

The detail-free algorithm becomes substantially simpler for the special case of two arms. For the sake of clarity, we provide a standalone exposition. The sampling stage is similar to Algorithm~\ref{alg:block}, except that the exploit arm is chosen using the sample average reward of arm $1$ instead of its posterior mean reward, and picking arm $2$ as the exploit arm only if it appears better by a sufficient margin. The racing stage is a simple ``race" between the two arms: it alternates the two arms until one of them can be eliminated with high confidence, regardless of the prior, and uses the remaining arm from then on.

Below we provide the algorithm and analysis for each stage separately. The corresponding lemmas can be used to derive
Theorem~\ref{thm:DF} for two arms, essentially same way as in Section~\ref{sec:DF-wrapup}.

Compared to the general case of $m>2$ arms, an additional provable guarantee is that even computing the threshold $N_\mP$ exactly does not require much information: it only requires knowing the prior mean rewards of both arms, and evaluating the CDFs for $\mu_1$ and $\mu_1-\mu_2$ at a single point each.

\subsection{The sampling stage for two arms}
\label{sec:DF-sampling-two}

The initial sampling algorithm is similar in structure to that in Section~\ref{sec:building-block}, but chooses the ``exploit" arm $a^*$ in a different way (after several samples of arm $1$ are drawn). Previously, $a^*$ was the arm with the highest posterior mean given the collected samples, whereas now the selection of $a^*$ is based only on the sample average $\sa{\mu}_1$ of the previously drawn samples of arm $1$.

We use $\sa{\mu}_1$ instead of the posterior mean of arm $1$, and compare it against the prior mean $\mu_2^0$. (Note that $\mu_2^0$ is also the posterior mean of arm $2$, because the prior is independent across arms.) We still need
    $ \E[ \mu_i -\mu_j| a^*=i]$
to be positive for all arms $i,j$, even though $\sa{\mu}_1$ is only an imprecise estimate of the posterior mean of $\mu_1$. For this reason, we pick $a^*=2$ only if $\mu_2^0$ exceeds $\sa{\mu}_1$ by a sufficient margin.  Once $a^*$ is selected, it is recommended in all of the remaining rounds, except for a few rounds chosen u.a.r. in which arm $2$ is recommended. See Algorithm~\ref{alg:dfblock} for the details.

\begin{algorithm}[ht]
\SetKwInOut{Input}{Parameters}\SetKwInOut{Output}{Output}
\Input{$k,L\in\N$ and $C\in(0,1)$.}
\BlankLine
\nl Let $\sa{\mu}_1$ be the sample average of the resp. rewards\;
\nl {\bf if} $\sa{\mu}_1 \leq \mu_2^0-C$ {\bf then} $a^*=2$ {\bf else} $a^*=1$\;
\nl From the set $P$ of the next $L\cdot k$ agents, pick a set $Q$ of $k$ agents uniformly at random\;
\nl Every agent $p\in P-Q$ is recommended arm $a^*$\;
\nl Every agent $p\in Q$ is recommended arm $2$
\caption{The sampling stage: samples both arms $k$ times.}
\label{alg:dfblock}
\end{algorithm}

We prove that the algorithm is BIC as long as $C\approx \tfrac23\,\mu_2^0$, and parameters $k^*,L$ are larger than some prior-dependent thresholds. We write out these thresholds explicitly, parameterizing them by the ratio $\lambda = C/\mu_2^0$.

\begin{lemma}\label{lem:dfblock-ic}
Consider BIC bandit exploration with two arms.
Algorithm $\dfblock{k}$ with parameters $(C,k^*,L)$ completes in
    $Lk+\max(k,k^*)$
rounds. The algorithm is BIC if \propref{prop:DF} holds and $\mu_2^0>0$, and the parameters satisfy
\begin{align}
\lambda &\triangleq C/\mu_2^0 \in (0,\tfrac23), \nonumber\\
k^*     &\geq 2\; (\lambda\cdot\mu_2^0)^{-2}\; \log \tfrac{4}{\beta(\lambda)},
    \label{eq:lem:dfblock-ic:0}\\
L       &\geq 1+\tfrac{8}{\beta(\lambda)} \left(\mu_1^0-\mu_2^0\right), \nonumber
\end{align}
where
    $\beta(\lambda) = \lambda\cdot \mu_2^0 \cdot \Pr{\mu_1/\mu_2^0 \leq 1-\tfrac{3\lambda}{2}}$.
\end{lemma}

\begin{proof}
Let $Z=\mu_2-\mu_1$. Consider an agent $p>k$.
As in the proof of Lemma \ref{lem:block-ic} it suffices to show that
    $\E[Z|\rec{2}{p}]\; \Pr{\rec{2}{p}}\geq 0$.

Denote
    $\Ev_1=\{\mu_1\leq \mu_2^0-C\}$.
As in the proof of Lemma \ref{lem:block-ic},
\begin{align*}
\E[Z|\rec{2}{p}]\Pr{\rec{2}{p}} = \E[Z | \Ev_1] \Pr{\Ev_1}\cdot \left(1-\tfrac{1}{L}\right) + \E[Z]\cdot \tfrac{1}{L}.
\end{align*}
Thus for the algorithm to be incentive compatible we need to pick $L$:
\begin{align}\label{eq:lem:dfblock-ic:1}
L \geq 1+\frac{\mu_1^0-\mu_2^0}{\E[Z|\Ev_1]\Pr{\Ev_1}}.
\end{align}

It remains to lower-bound the quantity $\E[Z|\Ev_1]\;\Pr{\Ev_1}$. Since the means $\mu_i$ are independent,
\begin{equation*}
\E[Z|\Ev_1]\Pr{\Ev_1}=
    \E[\mu_2^0-\mu_1 | \mu_2^0-\sa{\mu}_1\geq C]\; \Pr{\mu_2^0-\sa{\mu}_1\geq C}.
\end{equation*}
Denote $X=\mu_2^0-\mu_1$ and $\sa{x}^k = \mu_2^0-\mu_1^k$. Observe that $\sa{x}_k$ is the sample average of $k$ i.i.d. samples from some distribution with mean $X$. Then the quantity that we want to lower bound is:
\begin{equation}\label{eq:lem:dfblock-ic:2}
\E[Z|\Ev_1]\;\Pr{\Ev_1}=\E[X| \sa{x}_k\geq C] \;\Pr{\sa{x}_k\geq C}.
\end{equation}
We will use Chernoff-Hoeffding Bound to relate the right-hand side with quantities that are directly related to the prior distribution $\mP$. More precisely, we use a corollary: Lemma~\ref{lem:lower}, which we state and prove in Appendix~\ref{app:chernoff}.

If we apply Lemma~\ref{lem:lower} with $C=\lambda\cdot \mu_2^0$ and $\zeta=\kappa=\frac{1}{2}$, then for
    $k\geq 2  (\lambda\cdot\mu_2^0)^{-2}\; \log \tfrac{4}{\beta(\lambda)}$
the right-hand side of \refeq{eq:lem:dfblock-ic:2} is at least $\tfrac18 \beta(\lambda)$. Therefore,
    $\E[Z|\Ev_1]\;\Pr{\Ev_1} \geq \tfrac18 \beta(\lambda)$.
Plugging this back into \refeq{eq:lem:dfblock-ic:1}, it follows that the conditions \eqref{eq:lem:dfblock-ic:0} suffice to guarantee BIC.
\end{proof}


\subsection{The racing stage for two arms}
\label{sec:DF-racing-two}

The racing phase alternates both arms until one of them can be eliminated with high confidence. Specifically, we divide time in phases of two rounds each, in each phase select each arm once, and after each phase check whether
    $|\sa{\mu}_1^n-\sa{\mu}_2^n|$,
is larger than some threshold, where $\sa{\mu}_i^n$ is the sample average of arm $i$ at the beginning of phase $n$. If that happens then the arm with the higher sample average ``wins the race", and we only pull this arm  from then on. The threshold needs to be sufficiently large to ensure that the ``winner" is indeed the best arm with very high probability, regardless of the prior. The pseudocode is in Algorithm~\ref{alg:arms-elim}.

\begin{algorithm}[ht]
\SetKwInOut{Input}{Input}\SetKwInOut{Output}{Output}
\Input{parameters $k\in\N$ and $\theta\geq 1$; time horizon $T$.}
\Input{$k$ samples of each arm $i$, denoted $r_i^1 \LDOTS r_i^k$.}
\BlankLine
\nl Let $\hat{\mu}_i^k= \frac{1}{k}\sum_{t=1}^{k}r_i^t$ be the sample average for each arm $i$\;
\nl Split remainder into consecutive phases of two rounds each,
    starting from phase $n=k$\;
\While{$|\hat{\mu}_1^n-\hat{\mu}_2^n|\leq \sqrt{\frac{\log(T\theta)}{n}}$}{
	\nl The next two agents are recommended both arms sequentially\;
	\nl Let $r_i^n$ be the reward of each arm $i=1,2$ in this phase,
        and $\hat{\mu}_i^{n+1} = \frac{1}{n+1}\sum_{t=1}^{n}r_i^t$\;
	\nl $n=n+1$\;
}
\nl For all remaining agents recommend $a^*=\max_{a\in \{1,2\}} \hat{\mu}_a^n$
\caption{BIC race for two arms.}
\label{alg:arms-elim}
\end{algorithm}

To guarantee BIC, we use two parameters: the number of samples $k$ collected in the sampling stage, and parameter $\theta$ inside the decision threshold. The $k$ should be large high enough so that when an agent sees a recommendation for arm $2$ in the racing stage there is a significant probability that it is due to the fact that arm $2$ has ``won the race" rather than to exploration. The $\theta$ should be large enough so that the arm that has ``won the race" would be much more appealing than the losing arm according to the agents' posteriors.

\begin{lemma}[BIC]\label{lem:elimination-ic}
Assume two arms.
Fix an absolute constant $\tau\in (0,1)$ and let
    $\theta_\tau =  \frac{4}{\tau} / \Pr{\mu_2-\mu_1\geq \tau} $.
Algorithm \ref{alg:arms-elim} is BIC if \propref{prop:DF} holds, and the parameters satisfy
    $\theta\geq \theta_\tau$
and
    $k\geq \theta^2_\tau \log T$.
\end{lemma}

The regret bound does not depend on the parameters; it is obtained via standard techniques.

\begin{lemma}[Regret]\label{lm:reg-bound}
Algorithm \ref{alg:arms-elim} with any parameters $k\in\N$  and $\theta\geq 1$ achieves ex-post regret
\begin{align}\label{eq:lm:reg-bound-1}
R(T) \leq \frac{8 \log(T\theta)}{|\mu_1-\mu_2|}.
\end{align}
Denoting
    $\Delta = |\mu_1-\mu_2|$
and letting $n^*$ be the duration of the initial sampling stage, the ex-post regret for each round $t$ of the entire algorithm (both stages) is
\begin{align}\label{eq:lm:reg-bound-2}
R(t)
    &\leq \min\left( n^*\Delta + \frac{8 \log(T\theta)}{\Delta},\; t\Delta \right)
    \leq \min\left(2n^*,\; \sqrt{8t\log (T\theta)} \right).
\end{align}
\end{lemma}

In the remainder of this section we prove Lemma~\ref{lem:dfblock-ic} and Lemma~\ref{lem:elimination-ic}. Denote
    $\hat{z}_n=\hat{\mu}_2^n-\hat{\mu}_1^n$
and let
    $c_n = \sqrt{\frac{\log(T\theta)}{n}}$
be the decision threshold for each phase $n \geq k$ in Algorithm \ref{alg:arms-elim}. For ease of notation we will assume that even after the elimination at every iteration a sample of the eliminated arm is also drawn, but simply not revealed to the agent. Let $Z=\mu_2-\mu_1$. According to Chernoff-Hoeffding Bound, the following event happens with high probability:
\begin{align}\label{eq:lem:elimination-ic-chernoff}
\mC = \left\{ \forall n\geq k: |Z-\hat{z}_n| < c_n \right\}.
\end{align}
We make this precise in the following claim:

\begin{claim}\label{cl:arms-elim-chernoff}
    $\Pr{\neg \mC | Z} \leq \frac{1}{T\theta}$.
\end{claim}
\begin{proof}
By Chernoff-Hoeffding Bound and the union bound, for any $Z$, we have:
\begin{align*}
\Pr{\neg \mC | Z}
    \leq \sum_{n\geq n^*} \Pr{|Z-\hat{z}_n|\geq c_n}
    \leq \sum_{n\geq n^*} e^{-2n\cdot c_n^2}
    \leq T\cdot e^{- 2\log(T\theta)}
    \leq \frac{1}{T\theta}. \qquad\qedhere
\end{align*}
\end{proof}

\noindent Generally we will view $\neg \mC$ as a ``bad event", show that the expected ``loss" from this event is negligible, and then assume that the ``good event" $\mC$ holds.

\begin{proof}[Proof of Lemma~\ref{lem:elimination-ic}]
Let $n^* = \theta^2_\tau \log T$. Fix phase $n\geq n^*$, and some agent $p$ in this phase.
We will show that
    $\E[Z|\rec{2}{p}]\; \Pr{\rec{2}{p}}\geq 0$.

By Claim~\ref{cl:arms-elim-chernoff}, we have
$\Pr{\neg \mC | Z}
    \leq \frac{1}{T\theta}
    \leq 1/\theta_\tau$.
Therefore, since $Z\geq -1$, we can write:
\begin{align*}
\E[Z|\rec{2}{p}]\Pr{\rec{2}{p}}  =~& \E[Z|\rec{2}{p}, \mC]\Pr{\rec{2}{p}, \mC} + \E[Z|\rec{2}{p}, \neg \mC]\Pr{\rec{2}{p},\neg \mC}\\
\geq~&  \E[Z|\rec{2}{p}, \mC]\Pr{\rec{2}{p}, \mC}  - \theta_\tau^{-1}.
\end{align*}

Now we will upper bound the first term. We will split the first integral into cases for the value of $Z$. Observe that by the definition of $c_n$ we have $c_{n^*} \leq \theta_\tau^{-1}$.  Hence, conditional on $\mC$, if $Z\geq \tau \geq  2 c_{n^*}$, then we have definitely stopped and eliminated arm $1$ at phase $n=k\geq n^*$, since $\hat{z}_n> Z- c_n \geq c_n$. Moreover, if $Z\leq - 2c_n$, then by similar reasoning, we must have already eliminated arm $2$, since $\hat{z}_n < c_n$. Thus in that case event $\rec{2}{p}$ cannot occur. Moreover, if we ignore the case when $Z\in [0,\tau)$, then the integral can only decrease, since the integrand is non-negative. Since we want to lower bound it, we will ignore this region. Hence:
\begin{align*}
\E[Z|\rec{2}{p}, \mC]\Pr{\rec{2}{p}, \mC}
\geq~& \tau \Pr{\mC, Z\geq \tau}-2\cdot c_n \Pr{\mC, -2c_n\leq Z\leq 0}\\
\geq~& \tau \Pr{\mC, Z\geq \tau}-2\cdot c_n\\
\geq~& \tau \Pr{\mC, Z\geq \tau}-\frac{\tau \cdot \Pr{Z\geq \tau}}{2}
\end{align*}
Using Chernoff-Hoeffding Bound, we can lower-bound the probability in the first summand of the above:
\begin{align*}
\Pr{\mC, Z\geq \tau} =~& \Pr{\mC|Z\geq \tau} \Pr{Z\geq \tau}
\geq \left(1-\sum_{t\geq n^*} e^{-2 \cdot c_n^2}\right)\cdot  \Pr{Z\geq \tau}\\
\geq~& \left(1-\tfrac{1}{T}\right)\cdot  \Pr{Z\geq \tau} \geq \tfrac{3}{4} \Pr{Z\geq \tau}.
\end{align*}
Combining all the above inequalities yields that $\E[Z|\rec{2}{p}]\Pr{\rec{2}{p}}\geq 0$.
\end{proof}

\begin{proof}[Proof of Lemma~\ref{lm:reg-bound}]
The ``Chernoff event" $\mC$ has probability at most $1-\tfrac{1}{T}$ by Claim~\ref{cl:arms-elim-chernoff}, so the expected ex-post regret conditional on $\neg C$ is at most $1$. For the remainder of the analysis we will assume that $\mC$ holds.

Recall that
    $\Delta = |\mu_1-\mu_2|$.
Observe that
    $\Delta \leq |\hat{z}_n| + c_n $.
Therefore, for each phase $n$ during the main loop we have
    $\Delta \leq 2\cdot c_n$,
so
    $n\leq \frac{4 \log (T\theta)}{\Delta^2}$.
Thus, the main loop must end by phase
    $\frac{4 \log (T\theta)}{\Delta^2}$.

In each phase in the main loop, the algorithm collects regret $\Delta$ per round. Event $\mC$ guarantees that the ``winner" $a^*$ is the best arm when and if the condition in the main loop becomes false, so no regret is accumulated after that. Hence, the total regret is
    $R(T)\leq \frac{8 \log (T\theta)}{\Delta}$,
as claimed in \refeq{eq:lm:reg-bound-1}.

The corollary \eqref{eq:lm:reg-bound-2} is derived exactly as the corresponding corollary in Lemma~\ref{lm:DF-regret-body}
\end{proof}

\subsection{A Chernoff bound for the proof of Lemma~\ref{lem:dfblock-ic}}
\label{app:chernoff}

In this section we state and prove a version of Chernoff bound used in the proof of Lemma~\ref{lem:dfblock-ic}.


\begin{lemma}\label{lem:lower}
Let $X\geq 1$ be a random variable. Let $x_1,\ldots, x_k \in [0,1]$ be i.i.d. random variables with $\E[x_i|X]=X$. Let $\hat{x}_k=\frac{1}{k}\sum_{t=1}^{k} x_t$. Then for any $C, \zeta, \kappa \in (0,1)$, if
\begin{equation*}
k\geq \frac{-\log\left(\kappa\cdot (1-\zeta)\cdot C\cdot \Pr{X\geq (1+\zeta)C}\right)}{2\zeta^2\cdot C^2},
\end{equation*}
then:
\begin{equation*}
\E[X| \hat{x}_k\geq C]\cdot\Pr{\hat{x}_k\geq C} \geq (1-\zeta)\cdot C\cdot \left(1-\kappa-\kappa(1-\zeta)\right)\cdot \Pr{X\geq (1+\zeta)C}.
\end{equation*}
\end{lemma}

We first prove the following Lemma.

\begin{lemma}
Let $x_1,\ldots, x_k$ be i.i.d. random variables with $\E[x_i|X]=X$, and $X\geq -1$. Let $\hat{x}_k=\frac{1}{n}\sum_{t=1}^{k} x_t$. Then for any $C>0$, $\epsilon\leq C$ and $\delta>0$:
\begin{equation}
\E[X| \hat{x}_k\geq C]\cdot\Pr{\hat{x}_k\leq C} \geq (C-\epsilon)\cdot \left(1-e^{-2\cdot \delta^2\cdot k}\right)\cdot \Pr{X\geq C+\delta} - e^{-2\cdot \epsilon^2\cdot k}
\end{equation}
\end{lemma}

\begin{proof}
Let $\Ev_{C-\epsilon}$ be the event $X\geq C-\epsilon$ and $\hat{\Ev}_{C}$ be the event that $\hat{x}_k\geq C$. By the Chernoff bound observe that:
\begin{equation}
\Pr{\neg \Ev_{C-\epsilon}, \hat{\Ev}_{C}} \leq \Pr{|X-\hat{x}_k|> \epsilon} \leq e^{-2\cdot \epsilon^2 \cdot k}
\end{equation}
Moreover, again by Chernoff bounds and a simple factorization of the probability:
\begin{align*}
\Pr{\Ev_{C-\epsilon}, \hat{\Ev}_{C}} \geq~& \Pr{ \Ev_{C+\delta},\hat{\Ev}_{C}}
= \PrC{\hat{\Ev}_{C} }{ \Ev_{C+\delta}} \cdot \Pr{\Ev_{C+\delta}} \\
=~& \left(1-\PrC{\hat{x}_k < C}{X\geq C+\delta}\right) \cdot \Pr{X\geq C+\delta}\\
\geq~& \left(1- \PrC{\hat{x}_k-X<\delta}{X\geq C+\delta}\right)  \cdot \Pr{X\geq C+\delta}\\
\geq~& \left(1-e^{-2\cdot\delta^2\cdot k}\right)\cdot \Pr{X\geq C+\delta}
\end{align*}
We can now break apart the conditional expectation in two cases and use the latter lower and upper bounds on the probabilities of each case:
\begin{align*}
\E[X|\hat{\Ev}_{C}]\Pr{\hat{\Ev}_{C}}=~& \E[X| \Ev_{C-\epsilon}, \hat{\Ev}_{C}]\Pr{\Ev_{C-\epsilon}, \hat{\Ev}_{C}}+\E[X| \neg \Ev_{C-\epsilon}, \hat{\Ev}_{C}]\cdot \Pr{ \neg \Ev_{C-\epsilon}, \hat{\Ev}_{C}}\\
\geq~& (C-\epsilon) \left(1-e^{-2\cdot\delta^2\cdot k}\right)\cdot \Pr{X\geq C+\delta}- 1\cdot e^{-2\cdot\epsilon^2\cdot k}
\end{align*}
Where we also used that $X\geq -1$. This concludes the proof of the lemma.
\end{proof}

\begin{proof}[Proof of Lemma \ref{lem:lower}]
Now, if $k\geq \frac{-\log\left(\kappa\cdot (1-\zeta)\cdot C\cdot \Pr{X\geq (1+\zeta)C}\right)}{2\zeta^2\cdot C^2}$ and appplying the above Lemma for $\delta=\epsilon = \zeta\cdot C$ we get:
\begin{align*}
\E[X|\hat{\Ev}_{C}]\Pr{\hat{\Ev}_{C}}\geq~& (1-\zeta) C  (1- \kappa\cdot (1-\zeta))\Pr{X\geq (1+\zeta)C} - \kappa(1-\zeta) C \Pr{X\geq (1+\zeta)C}\\
\geq~& (1-\zeta) C (1-\kappa-\kappa(1-\zeta)) \Pr{X\geq (1+\zeta)C}
\end{align*}
Which concludes the proof of the Lemma.
\end{proof}

\end{APPENDICES}

\end{document}